\documentclass[12pt,english]{article}

\usepackage{algorithm}
\usepackage{booktabs}
\usepackage{graphicx}
\usepackage{amsmath}
\usepackage{amssymb}
\usepackage{latexsym}
\usepackage{crop}
\usepackage{algorithmic,algorithm}
\usepackage{multirow}
\usepackage{bm}
\usepackage{bbm}
\usepackage{enumerate}
\usepackage{url}
\usepackage{array}
\usepackage{paralist}
\usepackage{diagbox}
\usepackage{wasysym}
\usepackage{booktabs}
\usepackage[dvipsnames]{xcolor}
\usepackage[colorlinks = true, pdfstartview = FitV, linkcolor = blue, citecolor = blue, urlcolor = blue]{hyperref}

\usepackage{listings}

\usepackage{rotating}

\usepackage[capitalise]{cleveref}
\crefname{equation}{}{}
\crefname{figure}{Figure}{Figures}
\creflabelformat{equation}{\textup{(#2#1#3)}}
\crefname{assumption}{Assumption}{Assumptions}
\crefname{condition}{Condition}{Conditions}

\usepackage{xspace}

\renewcommand\th{\textsuperscript{th}\xspace}

\usepackage{fullpage}
\usepackage{multirow}

\usepackage[sort,nocompress]{cite}

\usepackage{arydshln}
\usepackage{enumitem}
\setlist[enumerate,1]{leftmargin=*,wide=0em, noitemsep,nolistsep, label = {\bfseries \arabic*.}}
\setlist[itemize,1]{leftmargin=*,wide=0em, noitemsep,nolistsep}

\usepackage{pifont}
\newcommand{\reals}{\mathbb{R}}

\newcommand {\range}  { {\textnormal{Range}} }

\makeatletter
\newcommand*{\transpose}{%
	{\mathpalette\@transpose{}}%
}
\newcommand*{\@transpose}[2]{%
	\raisebox{\depth}{$\m@th#1\intercal$}%
}
\makeatother

\renewcommand {\AA}  { {\bm{A}} }

\newcommand {\HH}  { {\bm{H}} }
\newcommand {\VV}  { {\bm{V}} }
\newcommand {\WW}  { {\bm{W}} }

\newcommand {\QQ}  { {\bm{Q}} }
\newcommand {\bSS}  { {\bm{S}} }

\newcommand {\XX}  { {\bm{X}} }

\newcommand {\hY}  { {\widehat{Y}} }

\newcommand {\zz}  { {\bm z} }

\newcommand {\xx}  { {\bm x} }
\newcommand {\yy}  { {\bm y} }

\newcommand {\res}  { {\bm r} }

\newcommand {\vv}  { {\bm v} }
\newcommand {\ww}  { {\bm w} }

\newcommand {\bb}  { {\bm b} }

\newcommand {\ee}  { {\bm e} }

\newcommand{\defeq}{\mathrel{\mathop:}=}

\renewcommand{\Pr}{\hbox{\bf{Pr}}}

\newcommand*\xbar[1]{%
	\hbox{%
		\vbox{%
			\hrule height 0.5pt 
			\kern0.5ex
			\hbox{%
				\kern-0.1em
				\ensuremath{#1}%
				\kern-0.1em
			}%
		}%
	}%
} 

\definecolor{forestgreen}{rgb}{0.13, 0.55, 0.13}

\definecolor{amber}{rgb}{1.0, 0.75, 0.0}

\definecolor{bananayellow}{rgb}{.8, 0.6, 0}

\newcounter{comment}

\usepackage{amsthm}
\usepackage[framemethod=TikZ]{mdframed}

\mdfdefinestyle{theoremstyle}{%
	linewidth = 1pt,%
	roundcorner = 10pt,%
	leftmargin = 0,%
	rightmargin = 0,%
	backgroundcolor = cyan!3,%
	outerlinecolor = magenta!70!black,%
	splittopskip = \topskip,%
	ntheorem = true,%
}
\mdtheorem[style=theoremstyle]{claim}{Claim}

\newmdtheoremenv[%
linewidth = 1pt,%
roundcorner = 10pt,%
leftmargin = 0,%
rightmargin = 0,%
backgroundcolor = green!3,%
outerlinecolor = blue!70!black,%
splittopskip = \topskip,%
ntheorem = true,%
]{theorem}{Theorem}

\newmdtheoremenv[%
linewidth = 1pt,%
roundcorner = 10pt,%
leftmargin = 0,%
rightmargin = 0,%
backgroundcolor = green!3,%
outerlinecolor = blue!70!black,%
splittopskip = \topskip,%
ntheorem = true,%
]{corollary}{Corollary}

\newmdtheoremenv[%
linewidth = 1pt,%
roundcorner = 10pt,%
leftmargin = 0,%
rightmargin = 0,%
backgroundcolor = green!3,%
outerlinecolor = blue!70!black,%
splittopskip = \topskip,%
ntheorem = true,%
]{lemma}{Lemma}

\newmdtheoremenv[%
linewidth = 1pt,%
roundcorner = 10pt,%
leftmargin = 0,%
rightmargin = 0,%
backgroundcolor = yellow!3,%
outerlinecolor = blue!70!black,%
splittopskip = \topskip,%
ntheorem = true,%
]{definition}{Definition}

\newmdtheoremenv[%
linewidth = 1pt,%
roundcorner = 10pt,%
leftmargin = 0,%
rightmargin = 0,%
backgroundcolor = green!3,%
outerlinecolor = blue!70!black,%
splittopskip = \topskip,%
ntheorem = true,%
]{proposition}{Proposition}

\newmdtheoremenv[%
linewidth = 1pt,%
roundcorner = 10pt,%
leftmargin = 0,%
rightmargin = 0,%
backgroundcolor = green!3,%
outerlinecolor = blue!70!black,%
splittopskip = \topskip,%
ntheorem = true,%
]{condition}{Condition}

\newmdtheoremenv[%
linewidth = 1pt,%
roundcorner = 10pt,%
leftmargin = 0,%
rightmargin = 0,%
backgroundcolor = green!3,%
outerlinecolor = blue!70!black,%
splittopskip = \topskip,%
ntheorem = true,%
]{assumption}{Assumption}

\theoremstyle{definition}
\newmdtheoremenv[%
linewidth = 1pt,%
roundcorner = 10pt,%
leftmargin = 0,%
rightmargin = 0,%
backgroundcolor = blue!3,%
outerlinecolor = blue!70!black,%
splittopskip = \topskip,%
ntheorem = true,%
]{example}{Example}

\theoremstyle{definition}
\newmdtheoremenv[%
linewidth = 1pt,%
roundcorner = 10pt,%
leftmargin = 0,%
rightmargin = 0,%
backgroundcolor = red!3,%
outerlinecolor = blue!70!black,%
splittopskip = \topskip,%
ntheorem = true,%
]{remark}{Remark}

\usepackage{tikz}
\usepackage{xparse}

\NewDocumentCommand\DownArrow{O{2.0ex} O{black}}{%
	\mathrel{\tikz[baseline] \draw [<-, line width=0.5pt, #2] (0,0) -- ++(0,#1);}
}

\usepackage{listings} 

\definecolor{mygreen}{rgb}{0,0.6,0}
\definecolor{mygray}{rgb}{0.5,0.5,0.5}
\definecolor{mymauve}{rgb}{0.58,0,0.82}

\newcommand*\dotprod[1]{\left\langle #1\right\rangle}
\newcommand*\vnorm[1]{\left\| #1\right\|}

\newcommand*\prob[1]{\Pr\left( #1\right)}
\newcommand {\Ex} { {\mathbb E} }

\newcommand*\bigO[1]{\mathcal O\left( #1\right)}

\usepackage{dsfont}

\usepackage[latin1]{inputenc}
\usepackage{amsmath}
\usepackage{amsfonts}
\usepackage{amssymb}
\usepackage{makeidx}
\usepackage{graphicx}
\usepackage{caption}
\usepackage{subcaption}
\usepackage{framed}
\usepackage{booktabs,array}
\usepackage{xcolor}

\allowdisplaybreaks

\begin{document}
\title{LSAR: Efficient Leverage Score Sampling Algorithm for the Analysis of Big Time Series Data}
\author{
	Ali Eshragh\thanks{School of Information and Physical Sciences, University of Newcastle, Australia, and International Computer Science Institute, Berkeley, CA, USA. Email:  \tt{ali.eshragh@newcastle.edu.au}}
	\and
	Fred Roosta\thanks{School of Mathematics and Physics, University of Queensland, Australia, and International Computer Science Institute, Berkeley, CA, USA. Email:  \tt{fred.roosta@uq.edu.au}} 
	\and 
	Asef Nazari\thanks{School of Information Technology, Deakin University, Australia. Email:  \tt{asef.nazari@deakin.edu.au}} 
	\and 
	Michael W. Mahoney\thanks{Department of Statistics, University of California at Berkeley, USA, and International Computer Science Institute, Berkeley, CA, USA. Email: \tt{mmahoney@stat.berkeley.edu}}
}
\date{\today}
\maketitle


\begin{abstract}
    We apply methods from randomized numerical linear algebra (RandNLA) to develop improved algorithms for the analysis of large-scale time series data. We first develop a new fast algorithm to estimate the leverage scores of an autoregressive (\texttt{AR}) model in big data regimes. We show that the accuracy of approximations lies within $(1+\bigO{\varepsilon})$ of the true leverage scores with high probability. These theoretical results are subsequently exploited to develop an efficient algorithm, called \texttt{LSAR}, for fitting an appropriate \texttt{AR} model to big time series data. Our proposed algorithm is guaranteed, with high probability, to find the maximum likelihood estimates of the parameters of the underlying true \texttt{AR} model and has a worst case running time that significantly improves those of the state-of-the-art alternatives in big data regimes. Empirical results on large-scale synthetic as well as real data highly support the theoretical results and reveal the efficacy of this new approach.
\end{abstract}


\section{Introduction}
\label{sec:intro}

A \emph{time series} is a collection of random variables indexed according to the order in which they are observed in time. The main objective of \emph{time series analysis} is to develop a statistical model to forecast the future behavior of the system. At a high level, the main approaches for this include the ones based on considering the data in its original \emph{time domain} and those arising from analyzing the data in the corresponding \emph{frequency domain} \cite[Chapter 1]{Shu}). More specifically, the former approach focuses on modeling some future value of a time series as a parametric function of the current and past values by studying the correlation between adjacent points in time. The latter framework, however, assumes the primary characteristics of interest in time series analysis relate to periodic or systematic sinusoidal variations. Although the two approaches may produce similar outcomes for  many cases, the comparative performance is better done in the ``time domain'' \cite[Chapter 1]{Shu} which is the main focus of this paper.

Box and Jenkins \cite{Box} introduced their celebrated \emph{autoregressive moving average} (\texttt{ARMA}) model for analyzing stationary time series. Although it has been more than 40 years since this model was developed, due to its simplicity and vast practicality, it continues to be widely used in theory and practice. A special case of an \texttt{ARMA} model is an \emph{autoregressive} (\texttt{AR}) model, which merely includes the autoregressive component. Despite their simplicity, \texttt{AR} models have a wide range of applications spanning from genetics and medical sciences to finance and engineering (e.g., \cite{Hamilton1989,Anderson1998,Cha,She,Abo,Esh,Messner2019}). 

The main hyper-parameter of an \texttt{AR} model is its \emph{order}, which directly relates to the dimension of the underlying predictor variable. In other words, the order of an \texttt{AR} model amounts to the number of lagged values that are included in the model. In problems involving big time series data, selecting an appropriate order for an \texttt{AR} model amounts to computing the solutions of many potentially large scale \emph{ordinary least squares} (OLS) problems, which can be the main bottleneck of computations (cf.\ \cref{SecAR}). Here is where randomized sub-sampling algorithms can be used to greatly speed-up such model selection procedures.

For computations involving large matrices in general, and large-scale OLS problems in particular, randomized numerical linear algebra (RandNLA) has successfully employed various random sub-sampling and sketching strategies. There, the underlying data matrix is randomly, yet appropriately, ``compressed'' into a smaller one, while approximately retaining many of its original properties. As a result, much of the expensive computations can be performed on the smaller matrix; Mahoney \cite{mahoney2011randomized} and Woodruff \cite{woodruff2014sketching} provided an extensive overview of RandNLA subroutines and their many applications. Moreover, implementations of algorithms based on those ideas have been shown to beat state-of-the-art numerical routines (e.g., \cite{avron2010blendenpik,meng2014lsrn,yang2015implementing}). 

Despite their simplicity and efficient constructions, matrix approximations using \emph{uniform} sampling strategies are highly ineffective in the presence of non-uniformity in the data (e.g., outliers). In such situations, \emph{non-uniform} (but still i.i.d.) sampling schemes in general, and leverage score sampling in particular \cite{drineas2012fast}, are instrumental not only in obtaining the strongest worst case theoretical guarantees, but also in devising high-quality numerical implementations. In times series data, one might expect that sampling methods based on leverage scores can be highly effective (cf.\ \cref{fig:gas_LS_unif}). However, the main challenge lies in computing the leverage scores, which na\"{i}vely can be as costly as the solution of the original OLS problems. In this light, exploiting the structure of  the time series model for estimating the leverage scores can be the determining factor in obtaining efficient algorithms for time series analysis. We carry out that here in the context of \texttt{AR} models. In particular, our contributions can be summarized as follows: 

\begin{enumerate}[label = (\roman*)]
	\item We introduce RandNLA techniques to the analysis of big time series data. 
	\item By exploiting the available structure, we propose an algorithm for approximating the leverage scores of the underlying data matrix that is shown to be faster than the state-of-the-art alternatives. 
	
	\item We theoretically obtain a high-probability relative error bound on the leverage score approximations. 
	
	\item Using these approximations, we then develop a highly-efficient algorithm, called \texttt{LSAR}, for fitting \texttt{AR} models with provable guarantees. 
	
	\item We empirically demonstrate the effectiveness of the \texttt{LSAR} algorithm on several large-scale synthetic as well as real big time series data.
\end{enumerate}

The structure of this paper is as follows: \cref{Sec:Background} introduces \texttt{AR} models and RandNLA techniques in approximately solving large-scale OLS problems. \cref{Sec:TheoreticalResults} deals with the theoretical results on developing an efficient leverage score sampling algorithm to fit and estimate the parameters of an \texttt{AR} model. All proofs are presented in \cref{Sec:Proofs}. \cref{Sec:EmpiricalResults} illustrates the efficacy of the new approach by implementing it on several large-scale synthetic as well as real big time series data. \cref{Sec:Conclusion} concludes the paper and addresses future work.

\subsubsection*{Notation}

Throughout the paper, vectors and matrices are denoted by bold lower-case and bold upper-case letters, respectively (e.g., $\vv$ and $\VV$). All vectors are assume to be column vectors. We use regular lower-case to denote scalar constants (e.g., $ d $). Random variables are denoted by regular upper-case letters (e.g., $ Y $).
For a real vector, $ \vv $, its transpose is denoted by $ \vv^{\transpose} $. For two vectors $ \vv,\ww $, their inner-product is denoted as $ \dotprod{\vv, \ww}  = \vv^{\transpose} \ww$. For a vector $\vv$ and a matrix $\VV$, $\|\vv\|$ and $\|\VV\|$ denote vector $\ell_{2}$ norm and matrix spectral norm, respectively. The condition number of a matrix $ \AA $, which is the ratio of its largest and smallest singular values, is denoted by $ \kappa(\AA) $. Range of a matrix $ \AA \in \reals^{n \times d}$, denoted by $ \range(\AA) $, is a sub-space of $ \reals^{n} $, consisting all the vectors $ \left\{ \AA \xx \mid \xx \in \reals^{d} \right\}$. Adopting \texttt{Matlab} notation, we use $ \AA(i,:) $ to refer to the $ i\th $ row of the matrix $ \AA $ and consider it as a column vector. Finally, $ \ee_{i} $ denotes a vector whose $ i\th $ component is one, and zero elsewhere.


\section{Background}
\label{Sec:Background}

In this section, we present a brief overview of the two main ingredients of the results of this paper, namely autoregressive models (\cref{SecAR}) and leverage score sampling for OLS problems (\cref{Sec:RandNLA}).

\subsection{Autoregressive Models}
\label{SecAR}

A time series $\{Y_t;\, t=0,\pm 1,\pm 2,\ldots\}$ is called (weakly) stationary, if the mean $\Ex[Y_t]$ is independent of time $t$, and the auto-covariance $Cov(Y_{t},Y_{t+h})$ depends only on the lag $h$ for any integer values $t$ and $h$. A stationary time series $\{Y_t;\, t=0,\pm 1,\pm 2,\ldots\}$\footnote{Throughout this paper, we assume that $Y_t$'s are continuous random variables.} with the constant mean $\Ex[Y_t]=0$ is an \texttt{AR} model with the order $p$, denoted by $\mathtt{AR}(p)$, if we have
\begin{align}
\label{eq:arp}
\medskip Y_t & = \phi_1 Y_{t-1}+\cdots+\phi_p Y_{t-p}+W_t, 
\end{align}
where $\phi_p \neq 0$ and the time series $\{W_t;\, t=0,\pm 1,\pm 2,\ldots\}$ is a Gaussian white noise with the mean $\Ex[W_t] = 0$ and variance $Var(W_t) = \sigma_{W}^2$. Recall that a Gaussian white noise is a stationary time series in which each individual random variable $W_t$ has a normal distribution and any pair of random variables $W_{t_1}$ and $W_{t_2}$ for distinct values of $t_1,t_2\in\mathbb{Z}$ are uncorrelated. 

\begin{remark}
	For the sake of simplicity, we assume that $\Ex[Y_t] = 0$. Otherwise, if $\Ex[Y_t]=\mu\neq 0$, then one can replace $Y_t$ with $Y_t-\mu$ to obtain
	\begin{align*}
	\medskip Y_t-\mu & = \phi_1 (Y_{t-1}-\mu)+\cdots+\phi_p (Y_{t-p}-\mu) +W_t,
	\end{align*}
	which is simplified to 
	\begin{align*}
	\medskip Y_t & = \mu(1-\phi_1\cdots-\phi_p) + \phi_1 Y_{t-1}+\cdots+\phi_p Y_{t-p} +W_t.
	\end{align*}
\end{remark}

It is readily seen that each $\mathtt{AR}(p)$ model has $p+2$ unknown parameters consisting of the order $p$, the coefficients $\phi_i$ and the variance of white noises $\sigma_W^2$. Here, we briefly explain the common methods in the literature to estimate each of these unknown parameters.

\paragraph{Estimating the order $\bm{p}$.} A common method to estimate the order of an $\mathtt{AR}(p)$ model is to use the \emph{partial autocorrelation function} (PACF) \cite[Chapter 3]{Shu}. The PACF of a stationary time series $\{Y_t;\, t=0,\pm 1,\pm 2,\ldots\}$ at lag $h$ is defined by
\begin{align}
\mathtt{PACF}_h & \defeq \begin{cases}
\medskip \rho(Y_t,Y_{t+1}) & \mbox{for}\ h=1, \\
\medskip \rho(Y_{t+h}-\hY_{t+h,-h}, Y_t-\hY_{t,h}) & \mbox{for}\ h\geq 2,
\end{cases}
\label{eq:PACF}
\end{align}
where $\rho$ denotes the correlation function, and where $\hY_{t,h}$ and $\hY_{t+h,-h}$ denote the linear regression, in the population sense, of $Y_t$ and $Y_{t+h}$ on $\{Y_{t+1},\ldots, Y_{t+h-1}\}$, respectively. It can be shown that for a causal $\mathtt{AR}(p)$ model, while the theoretical PACF \cref{eq:PACF} at lags $h=1,\ldots,p-1$ may be non-zero and at lag $h=p$ may be strictly non-zero, at lag $h=p+1$ it drops to zero and then remains at zero henceforth \cite[Chapter 3]{Shu}. Recall that an $\mathtt{AR(p)}$ model is said to be \emph{causal} if the time series $\{Y_t;\, t=0,\pm 1,\pm 2,\ldots\}$ can be written as $Y_t = W_t+\sum_{i=1}^{\infty} \psi_i W_{t-i}$ with constant coefficients $\psi_i$ such that $\sum_{i=1}^{\infty}|\psi_i|<\infty$. Furthermore, if a sample of size $n$ is obtained from a causal $\mathtt{AR}(p)$ model, then under some mild conditions, an estimated sample PACF at lags $h>p$, scaled by $\sqrt{n}$, has a standard normal distribution, in limit as $n$ tends to infinity \cite[Chapter 8]{Blackwell2009TimeSeries}. 

Thus, in practice, the sample PACF versus lag $h$ along with a $95\%$ zero-confidence boundary, that is two horizontal lines at ${\pm1.96}/{\sqrt{n}}$, are plotted. Then, the largest lag $h$ in which the sample PACF lies out of the zero-confidence boundary for PACF is used as an estimation of the order $p$. For instance, \cref{fig:PACF Exact AR20,fig:PACF Exact AR100,fig:PACF Exact AR200} display the sample PACF plots for the synthetic time series data generated from models $\mathtt{AR(20)}$, $\mathtt{AR(100)}$, and $\mathtt{AR(200)}$, respectively. Each figure illustrates that the largest PACF lying out of the red dashed $95\%$ zero-confidence boundary, locates at a lag which is equal to the order of the \texttt{AR} model.

\vspace*{-0.35cm}
\paragraph{Maximum likelihood estimation of the coefficients $\bm{\phi_i}$ and variance $\bm{\sigma_{W}^2}$.} Let $y_1,\ldots,y_n$ be a time series realization of an $\mathtt{AR}(p)$ model where $p$ is known and $n \gg p$. Unlike a linear regression model, the log-likelihood function
\[\log(f_{Y_1,\ldots,Y_n}(y_1,\ldots,y_n;\,\phi_1,\ldots,\phi_p,\sigma_{W}^2)),\] 
where $f$ is the joint probability distribution function of the random variables $Y_1,\ldots,Y_n$, is a complicated non-linear function of the unknown parameters. Hence, finding an analytical form of the maximum likelihood estimates (MLEs) is intractable. Consequently, one typically uses some numerical optimization methods to find an MLE of the parameters of an $\mathtt{AR(p)}$ model approximately. However, it can be shown that the conditional log-likelihood function is analogous to the log-likelihood function of a linear regression model given below \cite[Chapter~5]{Ham}:
\begin{align*}
\medskip & \log(f_{Y_{p+1},\ldots,Y_n|Y_1,\ldots,Y_p}(y_{p+1},\ldots,y_n | y_1,\ldots,y_p;\,\phi_1,\ldots,\phi_p,\sigma_{W}^2)) \\
\medskip & \hspace*{0.5cm} = -\frac{n-p}{2}\log(2\pi) - \frac{n-p}{2}\log(\sigma_{W}^2) - \sum_{t=p+1}^{n}\frac{(y_t-\phi_1 y_{t-1}-\cdots-\phi_p y_{t-p})^2}{2 \sigma_{W}^2}.
\end{align*}


Thus, the conditional MLE (CMLE) of the coefficients $\phi_i$ as well as the variance $\sigma_{W}^2$ can be estimated from an OLS regression of $y_t$ on $p$ of its own lagged values. More~precisely, 
\begin{align}
\label{eq:phi} \bm{\phi}_{n,p} & \defeq (\XX_{n,p}^{\transpose} \XX_{n,p})^{-1} \XX_{n,p}^{\transpose} \bm{y}_{n,p},
\end{align}
where $\bm{\phi}_{n,p}$ is the CMLE of the coefficient vector $[\phi_1,\ldots,\phi_p]^{\transpose}$, the data matrix
\begin{align}
\label{eq:Xnp}
\XX_{n,p} & \defeq \begin{pmatrix}
y_p & y_{p-1} & \cdots & y_{1} \\
y_{p+1} & y_{p} & \cdots & y_{2} \\ 
\vdots & \vdots & \ddots & \vdots \\	
y_{n-1} & y_{n-2} & \cdots & y_{n-p} 
\end{pmatrix},
\end{align}
and $\bm{y}_{n,p} \defeq \begin{bmatrix} y_{p+1} & y_{p+2} & \ldots & y_{n} \end{bmatrix}^{\transpose}$. 

\begin{remark}
    The data matrix $\XX_{n,p}$ in \cref{eq:Xnp} possesses Toeplitz structure that we take advantage of for our derivations in this paper, in particular developing the recursion for the leverage scores given in \cref{thm:TheRec4LevScr}. Also, it is highlighted that as the estimated parameter vector \cref{eq:phi} is operating under ``conditional'' MLE, the data matrix $\XX_{n,p}$ is a fixed design matrix.
\end{remark}

Moreover, the CMLE of $\sigma_{W}^2$, the so-called MSE, is given~by
\begin{align}\label{eq:MSE}
\widehat{\sigma}_{W}^2 & = \frac{\vnorm{\res_{n,p}}^{2}}{n-p},
\end{align}
where
\begin{align}
\label{eq:res} \res_{n,p} & \defeq \bm{y}_{n,p} - \XX_{n,p}\bm{\phi}_{n,p} \ =\ \begin{bmatrix}
\res_{n,p}(1) & \ldots & \res_{n,p}(n-p)
\end{bmatrix}^{\transpose}
\end{align}
and
\begin{align*}
\res_{n,p}(i) & = y_{p+i} - \XX_{n,p}^{\transpose}(i,:)\bm{\phi}_{n,p}\ \ \mbox{for}\ i=1, \ldots, n-p.
\end{align*}
Recall that $\XX_{n,p}(i,:)$ is the $ i\th $ row of matrix $ \XX_{n,p} $, that is, 
\begin{align*}
\XX_{n,p}(i,:) & \defeq \begin{bmatrix} y_{i+p-1} & y_{i+p-2} & \dots & y_i  \end{bmatrix}^{\transpose}. 
\end{align*}

One may criticize the CMLE as it requires one to exclude the first $p$ observations to construct the conditional log-likelihood function. Although this is a valid statement, due to the assumption $n \gg p$, dropping the first $p$ observation from the whole time series realization could be negligible.  

\begin{remark} \label{rem:pacf}
	It can be shown \cite[Chapter 3]{Shu} that if 
	\begin{align*}
	\medskip \hY_{t+h,-h} & = \alpha_1 Y_{t+h-1} + \cdots + \alpha_{h-1} Y_{t+1}, 
	\end{align*}
	then 
	\begin{align*}
	\medskip \hY_{t,h} & = \alpha_1 Y_{t+1} + \cdots + \alpha_{h-1} Y_{t+h-1}.
	\end{align*}
	This implies that finding PACF at each lag requires the solution to only one corresponding OLS problem. Furthermore, one can see that an empirical estimation of the coefficients $\alpha_i$ is the same as finding a CMLE of the coefficients of an $\mathtt{AR}(h-1)$ model fitted to the data. Thus, empirically estimating the order $p$ using a given time series data involves \emph{repeated} solutions of OLS problems, which can be computationally prohibitive in large-scale settings. Indeed, for $ n $ realizations $ y_{1},\ldots, y_{n}$, PACF at lag $ h $ can be calculated in $ \bigO{n h} $ using Toeplitz properties of the underlying matrix, and as a result selecting an appropriate order parameter $ p $ amounts to $ \bigO{\sum_{h=1}^p n h} = \bigO{n p^{2}} $ time complexity.
\end{remark}

\begin{remark}
	It should be noted that there is another method to estimate the parameters of an $\mathtt{AR}(p)$ model by solving the Yule-Walker equations with the Durbin-Levinson algorithm \cite[Chapter 8]{Blackwell2009TimeSeries}. Although, those estimates have asymptotic properties similar to CMLEs, solving the corresponding OLS problem is computationally faster than the Durbin-Levinson algorithm and also the CMLEs are statistically more efficient. 
\end{remark}

\subsection{Leverage Scores and RandNLA}
\label{Sec:RandNLA}

Linear algebra, which is the mathematics of linear mappings between vector spaces, has long had a large footprint in statistical data analysis. For example, canonical linear algebra problems such as principal component analysis and OLS are arguably among the first and most widely used techniques by statisticians. In the presence of large amounts data, however, such linear algebra routines, despite their simplicity of formulation, can pose significant computational challenges. For example, consider an over-determined OLS problem
\begin{align}
\label{eq:ls}
\min_{\xx} \vnorm{\AA\xx - \bb}^{2},
\end{align}
involving $ m \times d $ matrix $ \AA $, where $ m > d $. Note that, instead of $ n - p $ and $ p $ for the dimensions of the matrix \cref{eq:Xnp}, we adopt the notation $ m $ and $ d $ for the number of rows and columns, respectively. This is due to the fact that our discussion in this section involves arbitrary matrices and not those specifically derived from \texttt{AR} models. 
Solving \cref{eq:ls} amounts to $ \bigO{m d^{2} + d^{3}/3}$ flops by forming the normal equations, $ \bigO{m d^{2} - d^{3}}$ flops via QR factorization with Householder reflections, and $ \bigO{m d^{2} + d^{3}}$ flops using singular value decomposition (SVD) \cite{golub1983matrix}. Iterative solvers such as LSQR \cite{paige1982lsqr}, LSMR \cite{fong2011lsmr}, and LSLQ \cite{estrin2019lslq}, involve matrix-vector products at each iterations, which amount to $ \bigO{m d c} $ flops after $ c $ iterations. In other words, in ``big-data'' regimes where $ m d^2 \gg 1$, na\"{i}vely performing these algorithms can be costly. 

RandNLA subroutines involve the construction of an appropriate sampling/sketching matrix, $ \bSS \in \reals^{s \times m} $ for $ d \leq s \ll m $,  and compressing the data matrix into a smaller version $ \bSS \AA \in \reals^{s \times d}$. In the context of \cref{eq:ls}, using the smaller matrix, the above-mentioned classical OLS algorithms can be readily applied to the smaller scale problem
\begin{align}
\label{eq:sls}
\min_{\xx} \vnorm{\bSS\AA\xx - \bSS\bb}^{2},
\end{align}
at much lower costs. In these algorithms, sampling/sketching is used to obtain a data-oblivious or data-aware subspace embedding, which ensures that for  any $ 0 < \varepsilon,\delta < 1 $ and for large enough $ s $, we get
\begin{align}
\label{eq:LS_approx}
\prob{\vnorm{\AA\xx^{\star} - \bb}^{2} \leq \vnorm{\AA\xx_{s}^{\star} - \bb}^{2}  \leq \left(1+\bigO{\varepsilon}\right) \vnorm{\AA\xx^{\star} - \bb}^{2}} & \geq 1-\delta,
\end{align}
where $ \xx^{\star} $ and $ \xx_{s}^{\star} $ are the solutions to \cref{eq:ls,eq:sls}, respectively.
In other words, the solution to the reduced problem \cref{eq:sls} is a $ 1+\bigO{\varepsilon} $ approximation of the solution to the original problem~\cref{eq:ls}.

Arguably, the simplest data-oblivious way to construct the matrix $ \bSS $ is using uniform sampling, where each row of $ \bSS $ is chosen uniformly at random (with or without replacement) from the rows of the $ m \times m $ identity matrix. Despite the fact that the construction and application of such a matrix can be done in constant $ \bigO{1} $ time, in the presence of non-uniformity among the rows of $ \AA $, such uniform sampling strategies perform very poorly. In such cases, it can be shown that one indeed requires $ s \in \bigO{m} $ samples to obtain the above sub-space embedding property. 

To alleviate this significant shortcoming, data-oblivious sketching schemes involve randomly transforming the data so as to smooth out the non-uniformities, which in turn allows for subsequent uniform sampling in the randomly rotated space \cite{drineas2011faster}.  Here, the random projection acts as a preconditioner (for the class of random sampling algorithms), which makes the preconditioned data better behaved (in the sense that simple uniform sampling methods can be used successfully) (e.g., \cite{mahoney2011randomized,mahoney2016lecture}). With such sketching schemes, depending on the random projection matrix, different sample sizes are required, for instance, $ \bigO{d \log (1/\delta) / \varepsilon^{2}} $ samples for Gaussian projection, $ \bigO{d \log (d/\delta) / \varepsilon^{2}} $ samples for fast Hadamard-based transforms, and $ \bigO{d^{2} \text{poly}(\log (d/\delta)) / \varepsilon^{2}} $ samples using sparse embedding matrices. Woodruff \cite{woodruff2014sketching} provided a comprehensive overview of such methods and their extensions. 

Alternative to data-oblivious random embedding methods are data-aware sampling techniques, which by taking into account the information contained in the data, sample the rows of the matrix proportional to non-uniform distributions. Among many such strategies, those schemes based on \emph{statistical leverage scores} \cite{drineas2012fast} have not only shown to improve worst case theoretical guarantees of matrix algorithms, but also they are amenable to high-quality numerical implementations \cite{mahoney2011randomized}. Roughly speaking, the ``best'' random sampling algorithms base their importance sampling distribution on these scores and the ``best'' random projection algorithms transform the data to be represented in a rotated basis where these scores are approximately uniform. 

The concept of statistical leverage score has long been used in statistical regression diagnostics to identify outliers \cite{rousseeuw2011robust}. Given a data matrix $ \AA  \in \reals^{m \times d}$ with $ m \geq d $, consider any orthogonal matrix $ \QQ $ such that $ \range(\QQ) = \range(\AA) $. The $ i\th $ leverage score corresponding to $ i\th $ row of $ \AA $ is defined as
\begin{align*}
\ell(i) \defeq \vnorm{\QQ(i,:)}^{2}.
\end{align*}

It can be easily shown that this is well-defined in that the leverage score does not depend on the particular choice of the basis matrix $ \QQ $. Furthermore, the $ i\th $ leverage score boils down to the $ i\th $ diagonal entry of the \emph{hat} matrix, that is,
\begin{subequations}
	\begin{align}
	\label{eq:lev}
	\ell(i) & = \ee_{i}^{\transpose} \HH \ee_{i}\, \quad \mbox{for}\ i=1, \ldots, m,
	\end{align}
	where 
	\begin{align}
	\label{eq:hat}
	\HH \defeq \AA \left(\AA^{\transpose} \AA\right)^{-1}\AA^{\transpose}.
	\end{align}
	It is also easy to see that
	\begin{align*}
	\ell(i) \geq 0 \;\; \forall \; i, \quad \text{ and } \quad \sum_{i=1}^{m} \ell(i) = d.
	\end{align*}
	Thus,
	\begin{align}
	\label{eq:pi}
	\pi(i) \defeq \frac{\ell(i)}{d}, \quad \mbox{for}\ i=1, \ldots, m,
	\end{align}
\end{subequations}
defines a non-uniform probability distribution over the rows of $ \AA $. 

\paragraph{Leverage score sampling matrix $\bSS$.} Sampling according to the leverage scores amounts to randomly picking and re-scaling rows of $ \AA $ proportional to their leverage scores and appropriately re-scaling the sampled rows so as to maintain an unbiased estimator of $ \AA^{\transpose} \AA $, that is,
\begin{align*}
\Ex [\vnorm{\bSS \AA \xx}^{2}] = \vnorm{\AA \xx}^{2},\; \forall \xx.
\end{align*}
More precisely, each row of the $s \times m$ sampling matrix $\bSS$ is chosen randomly from the rows of the $m \times m$ identity matrix according to the probability distribution \cref{eq:pi}, with replacement. Furthermore, if the $i^{th}$ row is selected, it is re-scaled with the multiplicative factor 
\begin{align}
\label{rescaling-factor}
\frac{1}{\sqrt{s \pi_i}},
\end{align}
implying that $ 1/\sqrt{s \pi_i} \ee_{i}^{\transpose} $ is appended to $\bSS$.

Clearly, obtaining any orthogonal matrix $ \QQ $ as above by using SVD or QR factorization is almost as costly as solving the original OLS problem (i.e., $\bigO{m d^{2}} $ flops), which defeats the purpose of sampling altogether. In this light, Drineas et al. \cite{drineas2012fast} proposed randomized approximation algorithms, which efficiently estimate the leverage scores in $ \bigO{m d \log m + d^{3}} $ flops. For sparse matrices, this was further improved by Clarkson and Woodruff \cite{clarkson2017low}, Meng and Mahoney \cite{MM13_STOC}, and Nelson and Nguyen \cite{NN13} to $ \bigO{nnz(\AA) \log m + d^{3}} $. In particular, it has been shown that with the leverage score estimates $ \hat{\ell}(i) $ such that
\begin{align}
\label{eq:lev_beta}
\hat{\ell}(i) \geq \beta \ell(i), \quad \mbox{for}\ i = 1,2,\ldots m,
\end{align}
for some \emph{misestimation factor} $ 0 < \beta \leq 1 $, one can obtain \cref{eq:LS_approx} with
\begin{align}
\label{eq:lev_beta_sample}
s \in \bigO{d \log(d/\delta)/(\beta \varepsilon^{2})},
\end{align}
samples \cite{woodruff2014sketching}. As it can be seen from \cref{eq:lev_beta_sample}, the required sample size $ s $ is adversely affected by the leverage score misestimation factor $ \beta $.  

Recently, randomized sublinear time algorithms for estimating the parameters of an \texttt{AR} model for a given order $ d $ have been developed by Shi and Woodruff \cite{shi2019sublinear}. There, by using the notion of generalized leverage scroes, the authors propose a method for approximating CMLE of the parameters in $ \mathcal{O}(m \log^2 m+ (d^2\log^2 m)/\varepsilon^2 + (d^3 \log m)/\varepsilon^2) $ time, with high probability.
The analysis in \cite{shi2019sublinear} makes use of Toeplitz structure of data matrices arising from \texttt{AR} models. 
Also related to our settings here are \cite{van2003superfast} and \cite{xi2014superfast}, which developed, respectively, an exact and a (numerically stable) randomized approximation algorithm to solve Toeplitz least square problems, both with the time complexity of $ \bigO{ (m+d) \log^2(m+d) } $. 
An alternative sub-sampling algorithm to algorithmic leveraging for OLS problems has been considered by Wang  et al. \cite{wang2019information}. There, the sub-sampling is approached from the perspective of optimal design using D-optimality criterion, aiming to maximize the determinant of the Fisher information in the sub-sample. 
We also note that algorithms various statistical aspects of leverage scores have been extensively studied by Raskutti and Mahoney \cite{raskutti2016statistical} and Ma et al. \cite{ma2015statistical}. 
Finally, a more general notion of leverage scores in the context of recovery of continuous time signals from discrete measurements has recently been introduced by Avron et al. \cite{avron2019universal}.


\subsection{Theoretical Contributions}
\label{sec:contribution}

Here, by taking the advantage of the structure of \texttt{AR} models, we derive an algorithm, called \texttt{LSAR}, which given the (approximate) leverage scores of the data matrix for an $\mathtt{AR}(p-1)$ model (cf.\ \cref{eq:Xnp}), efficiently provides an estimate for the leverage scores related to an $\mathtt{AR}(p)$ model. In the process, we derive explicit bounds on the misestimation factor $ \beta $ in \cref{eq:lev_beta}.  An informal statement of our main results (\cref{thm:quasi,thm:main,thm:quality_assurance,thm:Time Complexity}) are as follows.

\paragraph{Claim (Informal).} For any $\varepsilon > 0$ small enough, we prove (with a constant probability of success): 
\begin{itemize}
	
	\item \cref{thm:main}: If only some suitable approximations of the leverage scores of an $ \mathtt{AR}(p-1) $ model are known, we can estimate those of an $ \mathtt{AR}(p) $ model with a misestimation factor $ \beta \in 1-\mathcal{O}(p\sqrt{\varepsilon}) $ in  $ \mathcal{O}(n + p^{3} \log p) $ time complexity. This should be compared with na\"{i}ve QR-based methods with $ \bigO{n p^2} $ and the universal approximation schemes developed by Drineas et al. \cite{drineas2012fast} with $ \bigO{n p \log n + p^{3}} $. 
		
	
	\item \cref{thm:quality_assurance,thm:Time Complexity}: Furthermore, an appropriate $ \mathtt{AR}(p) $ model can be fitted, with high-probability, in overall time complexity of $ \mathcal{O}(np + (p^{4} \log p)/\varepsilon^{2}) $ as compared with  $ \mathcal{O}(n p^{2}) $ using exact methods (cf.\ \cref{rem:pacf}), $\mathcal{O}((n+p)p\log^2(n+p) )$ by leveraging structured matrices as in \cite{van2003superfast}, and $ \mathcal{O}(n p \log^2 n+ (p^3\log^2 n)/\varepsilon^2 + (p^4 \log n)/\varepsilon^2) $ from sublinear time algorithms developed by Shi and Woodruff \cite{shi2019sublinear}.
\end{itemize}

\begin{remark}
	\label{rem:Conj}
	In big data regimes where typically $ n \gg p $ the above result implies an improvement over the existing methods for fitting an appropriate \texttt{AR} model. However, we believe that the dependence of the misestimation factor $\beta \in 1-\mathcal{O}(p \sqrt{\varepsilon})$ on $ p $ is superfluously a by-product of our analysis, as in our numerical experiments, we show that a sensible factor may be in the order of $\beta \in 1-\mathcal{O}(\log p \sqrt{\varepsilon})$. 
\end{remark}

\section{Theoretical Results}\label{Sec:TheoreticalResults}

In this section, we use the specific structure of the data matrix induced by an \texttt{AR} model to develop a fast algorithm to approximate the leverage scores corresponding to the rows of the data matrix \cref{eq:Xnp}. Furthermore, we theoretically show that our approximations possess relative error (cf.\ \cref{equ:MaxRelErrors}) bounds with high probability. Motivated from the leverage score based sampling strategy in \cref{Sec:RandNLA}, we then construct a highly efficient algorithm, namely \texttt{LSAR}, to fit an appropriate $\mathtt{AR}(p)$ model on big time series data. It should be noted that all proofs of this section are presented in \cref{Sec:Proofs}.

\subsection{Leverage Score Approximation for \texttt{AR} Models}

We first introduce \cref{def:notation} which relates and unifies notation of \cref{SecAR,Sec:RandNLA} together. 

\begin{definition} 
	\label{def:notation}
	In what follows, we define $ \ell_{n,p}$, $\HH_{n,p} $, and $ \pi_{n,p} $ as
	\begin{align*}
	\medskip \ell_{n,p}(i) &\defeq \ee_{i}^{\transpose} \HH_{n,p} \ee_{i} , \quad \mbox{for}\ i=1, \ldots, n-p, \\
	\medskip \HH_{n,p} &\defeq \XX_{n,p} \left(\XX_{n,p}^{\transpose} \XX_{n,p}\right)^{-1}\XX_{n,p}^{\transpose},\\
	\medskip \pi_{n,p}(i) &\defeq \frac{\ell_{n,p}(i)}{p}, \quad \mbox{for}\ i=1, \ldots, n-p.	
	\end{align*}
	That is, they refer, respectively, to  \cref{eq:lev},\cref{eq:hat}, and \cref{eq:pi}, using $ \AA = \XX_{n,p} $ as defined in \cref{eq:Xnp}.
\end{definition}

We show that the leverage scores associated with an $\mathtt{AR}(p)$ model can be recursively described using those arising from an $\mathtt{AR}(p-1)$ model. This recursive pattern is a direct result of the special structure of the data matrix \cref{eq:Xnp}, which amounts to a rectangular Hankel matrix \cite{golub1983matrix}.

\begin{theorem}[Exact Leverage Score Computations]
	\label{thm:TheRec4LevScr}
	The leverage scores of an $\mathtt{AR(1)}$ model are given by
	\begin{subequations}
		\begin{align}
		\label{eq:LS4AR(1)}\ell_{n,1}(i) & = \frac{y_i^2}{\displaystyle \sum_{t=1}^{n-1} y_t^2}, \quad \mbox{for}\ i=1,\ldots,n-1.
		\end{align} 
		For an $\mathtt{AR(p)}$ model with $ p \geq 2 $, the leverage scores are obtained by the following recursion
		\begin{align}
		\label{eq:LS4AR(p)}\ell_{n,p}(i) & = \ell_{n-1,p-1}(i) + \frac{\left(\res_{n-1,p-1}(i)\right)^{2}}{\vnorm{\res_{n-1,p-1}}^2}, \quad \mbox{for}\ i=1,\ldots,n-p,
		\end{align}
		where the residual vector $\res_{n-1,p-1}$ is defined in \cref{{eq:res}}.
	\end{subequations}
\end{theorem}

\cref{thm:TheRec4LevScr} shows that the leverage scores of \cref{eq:Xnp} can be exactly calculated through the recursive \cref{eq:LS4AR(p)} on the parameter $p$ with the initial condition \cref{eq:LS4AR(1)}. This recursion incorporates the leverage cores of the data matrix $\XX_{n-1,p-1}$ along with the residual terms of fitting an $\mathtt{AR(p-1)}$ model to the time series data $y_1,\ldots,y_{n-1}$. Note that both matrices $\XX_{n-1,p-1}$ and $\XX_{n,p}$ have equal number of rows, and accordingly equal number of leverage scores. Moreover, since we are dealing with big time series data (i.e., $n\gg p$), excluding one observation in practice is indeed negligible. 

\cref{thm:TheRec4LevScr}, though enticing at first glance, suffers from two major drawbacks in that not only does it require exact leverage scores associated with $ \mathtt{AR}(p-1) $ models, but it also involves exact residuals from the corresponding OLS estimations. In the presence of big data, computing either of these factors exactly defeats the whole purpose of data sampling altogether. To alleviate these two issues, we first focus on approximations in computing the latter, and then incorporate the estimations of the former. In doing so, we  obtain leverage score approximations, which enjoy desirable a priori relative error bounds. 

A natural way to approximate the residuals in the preceding $ \mathtt{AR}(p-1) $ model (i.e., $ \res_{n-1,p-1} $), is by means of sampling the data matrix $ \XX_{n-1,p-1} $ and solving the corresponding reduced OLS problem.
More specifically, we consider the sampled data matrix
\begin{subequations}
	\begin{align}
	\label{eq:samp_mat_exact}
	\tilde{\XX}_{n,p} \defeq \bSS \XX_{n,p},
	\end{align}
	where $ \bSS \in \reals^{s \times (n-p)} $ is the sampling matrix whose $ s $ rows are chosen at random with replacement from the rows of the $ (n-p) \times (n-p) $ identity matrix according to the distribution $ \{\pi_{n,p}(i)\}_{i=1}^{n-p} $ (cf.\ \cref{def:notation}) and rescaled by the appropriate factor \cref{rescaling-factor}.     
	Using $\tilde{\XX}_{n,p}$, the estimated parameter vector $\tilde{\bm{\phi}}_{n,p}$ is calculated as 
	\begin{align}
	\label{eq:phi_tilde}
	\tilde{\bm{\phi}}_{n,p} &\defeq (\tilde{\XX}_{n,p}^{\transpose} \tilde{\XX}_{n,p})^{-1} \tilde{\XX}_{n,p}^{\transpose} \tilde{\yy}_{n,p},
	\end{align}
	where $\tilde{\yy}_{n,p} \defeq \bSS \yy_{n,p}$. Finally, the residuals of $ \tilde{\bm{\phi}}_{n,p} $, analogous to \cref{eq:res}, are given by
	\begin{align}
	\label{eq:res_tilde}
	\tilde{\res}_{n,p} & \defeq \yy_{n,p} - \XX_{n,p}\tilde{\bm{\phi}}_{n,p}.    
	\end{align}
\end{subequations}

\vspace*{-0.05cm}
\begin{remark}
	\label{rem:res_tilde}
	We note that the residual vector $ \tilde{\res}_{n,p} $ is computed using the sampled data matrix $ \tilde{\XX}_{n,p} $, which is itself formed according to the leverage scores. In other words, the availability $ \tilde{\res}_{n,p} $ is equivalent to that of $ \{\pi_{n,p}(i)\}_{i=1}^{n-p} $.
\end{remark}

The following theorem, derived from the structural result  \cite{drineas2011faster}, gives estimates on the approximations \cref{eq:phi_tilde,eq:res_tilde}. 
\begin{theorem}[{\cite[Theorem 1]{drineas2011faster}}] 
	\label{thm:drma}
	Consider an $\mathtt{AR(p)}$ model and let $ 0 < \varepsilon,\delta < 1 $. For sampling with (approximate) leverage scores using a sample size $ s $ as in \cref{eq:lev_beta_sample} with $ d = p $, we have with probability at least $ 1-\delta $,
	\begin{subequations}
		\begin{align}
		\label{EquBoundSSE} \vnorm{\tilde{\res}_{n,p}} & \leq (1+\varepsilon) \vnorm{\res_{n,p}}, \\
		\label{EquRelDiffError} \vnorm{\bm{\phi}_{n,p}-\tilde{\bm{\phi}}_{n,p}} & \leq \sqrt{\varepsilon} \eta_{n,p} \vnorm{\bm{\phi}_{n,p}},
		\end{align}
	\end{subequations}
	where $ \bm{\phi}_{n,p}, \res_{n,p}, \tilde{\bm{\phi}}_{n,p} $ and $ \tilde{\res}_{n,p} $ are defined,, respectively, in \cref{eq:phi,eq:phi_tilde,eq:res,eq:res_tilde}, 
	\begin{align}
	\label{eq:eta}
	\eta_{n,p} = \kappa(\XX_{n,p})\sqrt{\xi^{-2}-1},
	\end{align}
	$\kappa(\XX_{n,p})$ is the condition number of matrix $\XX_{n,p}$, and $\xi\in(0,1]$ is the fraction of $\yy_{n,p} $ that lies in $ \range(\XX_{n,p}) $, that is, $\xi \defeq \vnorm{\HH_{n,p} \yy_{n,p}} / \vnorm{\yy_{n,p}}$ with $ \HH_{n,p} $ as in \cref{def:notation}.
\end{theorem}

Using a combination of exact leverage scores and the estimates \cref{eq:res_tilde} on the OLS residuals associated with the $ \mathtt{AR}(p-1) $ model, we define \emph{quasi-approximate leverage scores} for the $\mathtt{AR}(p)$ model.
\begin{definition}[Quasi-approximate Leverage Scores]
	\label{DefTildeLS}
	For an $\mathtt{AR(p)}$ model with $ p \geq 2 $, the quasi-approximate leverage scores are defined by the following equation
	\begin{align}
	\label{eq:lev_p_tilde}
	\medskip \tilde{\ell}_{n,p}(i) & \defeq \ell_{n-1,p-1}(i) + \frac{\left(\tilde{\res}_{n-1,p-1}(i)\right)^{2}}{\vnorm{\tilde{\res}_{n-1,p-1}}^2}\ \ \mbox{for}\ i=1,\ldots,n-p,
	\end{align}
	where $\ell_{n,p}(i)$ and $\tilde{\res}_{n,p} $ are as in \cref{def:notation,eq:res_tilde}.
\end{definition}



Clearly, the practical advantage of $ \tilde{\ell}_{n,p} $ is entirely contingent upon the availability of the exact leverage scores for $p-1$, that is, $ \ell_{n-1,p-1} $ (cf.\ \cref{DefTildeLS}). For $ p=2 $, this is indeed possible. More specifically, from \cref{eq:LS4AR(1)}, the exact leverage scores of an $\mathtt{AR(1)}$ model can be trivially calculated, which in turn give the quasi-approximate leverage scores $\{\tilde{\ell}_{n-2,2}(i)\}_{i=1}^{n-2}$ using \cref{eq:lev_p_tilde}. However, for $p = 3$ (and subsequent values), the relation \cref{eq:lev_p_tilde} does not apply as not only are $ \{\ell_{n-1,p-1}(i)\}_{i=1}^{n-p} $ no longer readily available, but also for the same token without having $ \{\pi_{n-1,p-1}(i)\}_{i=1}^{n-p} $, the residual vector $ \tilde{\res}_{n-1,p-1} $ may not be computed directly (cf.\ \cref{rem:res_tilde}). Nonetheless, replacing the exact leverage scores with quasi-approximate ones in \cref{eq:lev_p_tilde} for $ p=2 $ allows for a new approximation for $p = 3$. Such new leverage score estimates can be in turn incorporated in approximation of subsequent leverage scores for $ p \geq 4 $. This idea leads to our final and practical definition of \emph{fully-approximate leverage scores}.

\begin{definition}[Fully-approximate Leverage Scores]
	\label{DefHatLS}
	For an $\mathtt{AR(p)}$ model with $ p \geq 1 $, the fully-approximate leverage scores are defined by the following equation
	\begin{subequations}
		\label{eq:lev_full}
		\begin{align}
		\label{eq:lev_estimate}
		\hat{\ell}_{n,p}(i) & \defeq 
		\left\{\begin{array}{ll}
		\medskip \ell_{n,1}(i), & \mbox{for}\ p=1 \\
		\tilde{\ell}_{n,2}(i), & \mbox{for}\ p=2 \\
		\hat{\ell}_{n-1,p-1}(i) + \displaystyle{\frac{\left(\hat{\res}_{n-1,p-1}(i)\right)^{2}}{\vnorm{\hat{\res}_{n-1,p-1}}^2}},& \mbox{for}\ p\geq 3
		\end{array}\right.,
		\end{align}
		where 
		\begin{align}
		\label{eq:res_hat}
		\hat{\res}_{n-1,p-1} & \defeq \yy_{n-1,p-1} - \XX_{n-1,p-1}\hat{\bm{\phi}}_{n-1,p-1}, \\
		\label{eq:phi_hat}
		\hat{\bm{\phi}}_{n-1,p-1} &\defeq (\hat{\XX}_{n-1,p-1}^{\transpose} \hat{\XX}_{n-1,p-1})^{-1} \hat{\XX}_{n-1,p-1}^{\transpose} \hat{\yy}_{n-1,p-1}\,	
		\end{align}	
		and $\hat{\XX}_{n-1,p-1}$ and $\hat{\yy}_{n-1,p-1}$ are the reduced data matrix and response vector, sampled  respectively, according to the distribution 
		\begin{align}
		\label{eq:sampling_distribution}
		\hat{\pi}_{n-1,p-1}(i) & = \frac{\hat{\ell}_{n-1,p-1}(i)}{p-1}\ \ \mbox{for}\ i=1,\ldots,n-p.
		\end{align}
	\end{subequations}
\end{definition}

\begin{remark}\label{rem:compare quasi- with fully-}
	It should be noted that \cref{eq:lev_p_tilde} estimates the leverage scores of an $\mathtt{AR}(p)$ model, given the corresponding \emph{exact} values of an $\mathtt{AR}(p-1)$ model. This is in sharp contrast to \cref{eq:lev_estimate}, which recursively provides similar estimates without requiring any information on the exact values.
\end{remark}

Unlike the quasi-approximate leverage scores, the fully-approximate ones in \cref{DefHatLS} can be easily calculated for any given the parameter value $p \geq 1$. Finally, \cref{thm:main} provides a priori relative-error estimate on individual fully-approximate leverage scores. 

\begin{theorem}[Relative Errors for Fully-approximate Leverage Scores]
	\label{thm:main}
	For the fully-approximate leverage scores, we have with probability at least $ 1-\delta $,
	\begin{align*}
	\medskip \frac{|\ell_{n,p}(i)-\hat{\ell}_{n,p}(i)|}{\ell_{n,p}(i)} & \leq \left( 1 + 3 \eta_{n-1,p-1} \kappa^2(\XX_{n,p})\right)(p-1) \sqrt{\varepsilon}, \quad \mbox{for}\ i = 1,\ldots,n-p,
	\end{align*}
	recalling that $\delta$, $\eta_{n,p}, \kappa(\XX_{n,p})$, and $\varepsilon$ are as in \cref{thm:drma}. 
\end{theorem}

Although qualitatively descriptive, the bound in \cref{thm:main} is admittedly pessimistic and involves an overestimation factor that scales quadratically with the condition number of the data matrix, $ \kappa $, and linearly with the order of the \texttt{AR} model, $ p $. We conjecture that the linear dependence on $ p $ can be instead replaced with $ \log(p) $, which is supported by the experiment depicted in \cref{fig:Tthe4_ver}. We leave the investigation of ways to improve the upper-bound of \cref{thm:main} to future work.

\cref{thm:main} prescribes the misestimation factor $\beta$ (cf.\ \cref{eq:lev_beta}) for the fully-approximate leverage sores of an $\mathtt{AR}(p)$ model, stated in \cref{cor:lev_beta}. 

\begin{corollary}
	\label{cor:lev_beta}
	The misestimation factor $\beta$  for the fully-approximate leverage scores of an $\mathtt{AR}(p)$ model is $ 1-\mathcal{O}(p \sqrt{\varepsilon}) $.
\end{corollary}

\subsection{\texttt{LSAR} Algorithm for Fitting \texttt{AR} Models}

Based on these theoretical results, we introduce the \texttt{LSAR} algorithm, depicted in \cref{alg:leverage}, which is the first leverage score sampling algorithm to approximately fit an appropriate \texttt{AR} model to a given big time series data. The theoretical properties of  the \texttt{LSAR} algorithm are given in \cref{thm:quality_assurance,thm:Time Complexity}. 
\begin{algorithm}[h!]
	\caption{\texttt{LSAR}: Leverage Score Sampling Algorithm for Approximate \texttt{AR} Fitting}	
	\begin{algorithmic}
		\STATE \textbf{Input:} 
		\begin{itemize}[label={-}]
			\vspace{1mm}
			\item Time series data $\{y_1,\ldots,y_n\}$\,;
			\vspace{1mm}
			\item A relatively large value $\bar{p} \ll n$\,;
			\vspace{1mm}
			\item Constant parameters $0< \varepsilon < 1$ and $0< \delta_{0}<1$;
		\end{itemize}
		\vspace{1mm}
		\item \emph{Step 0}. Set $p=0$ and $m=n-\bar{p}$\,;
		\vspace{1mm}
		\WHILE {$ p < \bar{p} $} 
		\vspace{1mm}
		\STATE \emph{Step 1}. $p \leftarrow p + 1$ and $m \leftarrow m + 1$\,;
		\vspace{1mm}
		\STATE \emph{Step 2}. Estimate PACF at lag $p$, i.e., $\hat{\tau}_{p}$\,;
		\vspace{1mm}
		\STATE \emph{Step 3}. Compute the approximate leverage scores $\hat{\ell}_{m,p}(i)$ for $i=1,\ldots,m-p$ as in \cref{eq:lev_estimate}\,;
		\vspace{1mm}
		\STATE \emph{Step 4}. Compute the sampling distribution $\hat{\pi}_{m,p}(i)$ for $i=1,\ldots,m-p$ as in \cref{eq:sampling_distribution}\,;
		\vspace{1mm}
		\STATE \emph{Step 5}. Set $s$ as in \cref{eq:lev_beta_sample} by replacing $d$ with $p$, $\delta = \delta_{0}/p$, and $\beta$ with the bound given in \cref{cor:lev_beta}\,; 
		\vspace{1mm}
		\STATE \emph{Step 6}. Form the $s \times m$ sampling matrix $\bSS$ by randomly choosing $s$ rows of the corresponding identity matrix according to the probability distribution found in Step $4$, with replacement, and rescaling them with the factor \cref{rescaling-factor}\,;
		\vspace{1mm}
		\STATE \emph{Step 7}. Construct the sampled data matrix $\hat{\XX}_{m,p} = \bSS \XX_{m,p}$ and response vector $\hat{\yy}_{m,p} = \bSS \yy_{m,p}$\,; 
		\vspace{1mm}
		\STATE \emph{Step 8}. Solve the associated reduced OLS problem to estimate the parameters $\hat{\bm{\phi}}_{m,p}$ and residuals $\hat{\res}_{m,p}$ as in \cref{eq:phi_hat,eq:res_hat}, respectively\,;
		\vspace{1mm}
		\ENDWHILE
		\vspace{1mm}
		\STATE \emph{Step 9}. Estimate $p^*$ as the largest $p$ such that $ |\hat{\tau}_{p}| \geq {1.96}/{\sqrt{s}} $\,;
		\vspace{1mm}
		\STATE \textbf{Output:} Estimated order $p^*$ and parameters $ \hat{\bm{\phi}}_{n-\bar{p}+p^*,p^*} $.
	\end{algorithmic}
	\label{alg:leverage}
\end{algorithm}

\begin{remark}
	For the overall failure probability, recall that in order to get an accumulative success probability of $1 - \delta_{0}$ for $\bar{p}$ iterations, the per-iteration failure probability is set as $\delta  = 1 - \sqrt[\bar{p}]{1 - \delta_{0}} \in \Omega(\delta_{0} / \bar{p} )$. However, since this dependence manifest itself only logarithmically, it is of negligible consequence in overall complexity.
\end{remark}

The quality of the fitted model by the \texttt{LSAR} algorithm depends on two crucial ingredients, the order of the underlying \texttt{AR} model as well the accuracy of the estimated parameters. The latter is guaranteed by \cref{thm:drma}. For the former,  \cref{thm:quality_assurance} shows that for small enough $ \varepsilon $, the \texttt{LSAR} algorithm can estimate the same model order as that using the full data matrix. 

Let $\tau_{p}$ and $\hat{\tau}_{p}$ be the PACF values estimated using the CMLE of parameter vectors based on the full and sampled data matrices, $ \bm{\phi}_{n,p-1} $ and $ \hat{\bm{\phi}}_{n,p-1} $, respectively.

\begin{theorem}[\texttt{LSAR} Model-order Estimation]
	\label{thm:quality_assurance}
	Consider a causal $\mathtt{AR(p^*)}$ model and let $ 0 < \varepsilon,\delta < 1 $. For sampling with fully-approximate leverage scores using a sample size $ s $ as in \cref{eq:lev_beta_sample} with $ d = p^* $ and $ \beta $ as in \cref{cor:lev_beta} with $ p = p^*$, we have with probability at least $ 1-\delta $,
	\begin{subequations}
		\begin{align}
		\medskip |\hat{\tau}_{p}| & \geq |\tau_{p}| - c_{1} \sqrt{\varepsilon}, \hspace{14mm} \quad \mbox{for}\ p = p^*, \label{eq:tau_1}\\
		\medskip |\hat{\tau}_{p}| & \leq |\tau_{p}| +  c_{2} \sqrt{(p-1) \varepsilon}, \quad \mbox{for}\ p > p^*, \label{eq:tau_2}
		\end{align}
	\end{subequations}
	where $ c_{1} $ and $ c_{2} $ are bounded positive constants depending on a given realization of the model. 
\end{theorem}

\cref{thm:quality_assurance} implies that, when $ |\tau_{p^*}| \geq 1.96/\sqrt{n} $ and $ |\tau_{p}| \leq 1.96/\sqrt{n}$ for $p > p^* $,  with high probability, we are guaranteed to have $ |\hat{\tau}_{p^*}| \geq 1.96/\sqrt{n} - \mathcal{O}(\sqrt{\epsilon})$ and $ |\hat{\tau}_{p}| \leq 1.96/\sqrt{n} + \mathcal{O}(\sqrt{\epsilon})$ for $p > p^* $, respectively. In practice, we can consider a larger bandwidth of size $ 2 \times 1.96/\sqrt{s} $; see the experiments of \cref{Sec:EmpiricalResults}.

\cref{thm:Time Complexity} gives the overall running time of the \texttt{LSAR} algorithm.
\begin{theorem}[\texttt{LSAR} Computational Complexity]
	\label{thm:Time Complexity}
	The worst case time complexity of the \texttt{LSAR} algorithm for an input $\mathtt{AR(p^*)}$ time series data is $\mathcal{O} \left( np^* + p^{*^{4}} \log p^* /\varepsilon^{2} \right)$, with probability at least $1-\delta_{0}$ ($0< \delta_{0} < 1$) and the $p\th$ iteration of the algorithm has $\delta = \delta_{0}/p$, which appears in the log for each sample size.
\end{theorem}

\begin{remark}\label{rem:Time Complexity}
	We believe that the restriction on $ \varepsilon $ given by \cref{thm:Time Complexity} is highly pessimistic and merely a by-product of our proof techniques here. As evidenced by numerical experiments, e.g., \cref{fig:Tthe4_ver}, we conjecture that a  more sensible bound is $0 < \varepsilon \leq (\log p^*)^{-2}$; see also the discussion in the last paragraph of \cref{Sec:Background} and \cref{rem:Conj}. In fact, even the tight bounds on the sample size for RandNLA routines rarely manifest themselves in practice (e.g., \cite{roszas,mahoney2011randomized,mahoney2016lecture}). Guided by these observations, in our numerical experiments of \cref{Sec:EmpiricalResults}, we set our sample sizes at factions of the total data, e.g., $ s = 0.001 n $, even for small values of $ p^* $.
\end{remark}


\section{Empirical Results} 
\label{Sec:EmpiricalResults}

In this section, we present the performance of the \texttt{LSAR} algorithm on several synthetic as well as real big time series data. The numerical experiments are run in MATLAB R2018b on a 64-bit windows server with dual processor each at 2.20GHz with 128 GB installed RAM. 

The numerical results reveal the efficiency of the \texttt{LSAR} algorithm, as compared with the classical alternative using the entire data. More precisely, it is illustrated that by sampling only $0.1\%$ of the data, not only are the approximation errors kept significantly small, but also the underlying computational times are considerably less than the corresponding exact~algorithms. 

We present our numerical analysis in three subsequent sections. In \cref{Sec:EmpRes-LS}, we report the computational times as well as the quality of leverage score approximations \cref{eq:lev_estimate} on three synthetically generated data by running Steps 0-8 of the \texttt{LSAR} algorithm. Analogously, \cref{Sec:EmpRes-PACF} shows similar results for estimating PACF (i.e., the output of Step 2 in the \texttt{LSAR} algorithm). Finally, \cref{Sec:EmpRes-RealData} displays the performance of the \texttt{LSAR} algorithm on a real big time series data. It should be noted that all computational times reported in this section are in ``seconds''.

\subsection{Synthetic Data: Verification of Theory}
\label{Sec:EmpRes-LS}

We generate synthetic large-scale time series data with two million realizations from the models $\mathtt{AR(20)}$, $\mathtt{AR(100)}$, and $\mathtt{AR(200)}$. For each dataset, the leverage scores over a range of lag values (i.e., the variable $h$ in the \texttt{LSAR} algorithm) are calculated once by using the exact formula as given in \cref{def:notation}, and another time by estimating the fully-approximate leverage scores as defined in \cref{eq:lev_estimate}. The latter is computed by running Steps 0-8 of the \texttt{LSAR} algorithm with $s = 0.001n = 2000$. 

\cref{fig:LS_rel_err} displays and compares the quality and run time between the fast sampled randomized Hadamard transform (SRHT) approximation technique developed by Drineas et al. \cite{drineas2012fast} and \cref{eq:lev_full}. At each lag $p$, the maximum pointwise relative error (\texttt{MPRE}, for short) is defined by 
\begin{align}\label{equ:MaxRelErrors}
	\max_{1\leq i\leq n-p}\left\{\frac{|\hat{\ell}_{n,p}(i) - \ell_{n,p}(i)|}{\ell_{n,p}(i)}\right\}.
\end{align}

As displayed in \cref{fig:MaxRelErrAR(20),fig:MaxRelErrAR(100),fig:MaxRelErrAR(200)}, while the \texttt{MPRE} curves have sharp increase at the beginning and then quickly converge to an upper limit around $0.1670$ for fully-approximate leverage scores, the output of SRHT seems to converge around $3$. This demonstrates the high-quality of the fully-approximate leverage scores using only $0.1\%$ of the rows of the data matrix. More interestingly, \cref{fig:LSTimeAR(20),fig:LSTimeAR(100),fig:LSTimeAR(200)} demonstrate the computational efficiency of the fully-approximate leverage scores. In light of the inferior performance of SRHT, both in terms of the quality of approximation and also run time, in the subsequent experiments, we will no longer consider SRHT approximation alternative.
\begin{figure}[h!]
	\centering
	\begin{subfigure}{.3\textwidth}
		\includegraphics[width=\textwidth]{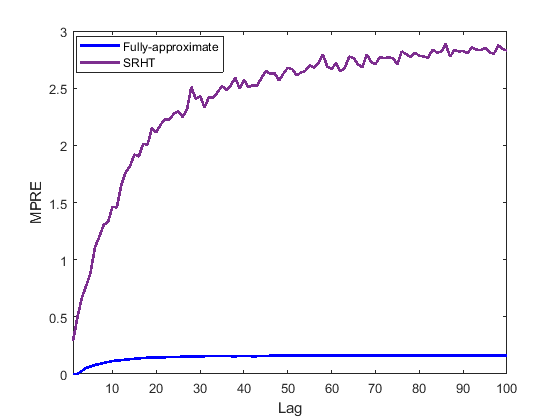}
		\caption{$\mathtt{AR(20)}$ }
		\label{fig:MaxRelErrAR(20)}
	\end{subfigure}
	\begin{subfigure}{.3\textwidth}
		\includegraphics[width=\textwidth]{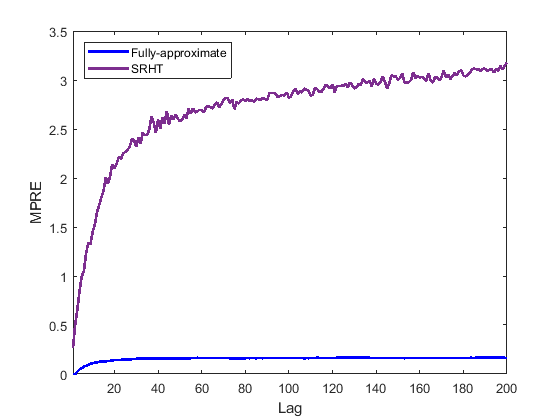}
		\caption{$\mathtt{AR(100)}$ }
		\label{fig:MaxRelErrAR(100)}
	\end{subfigure}
	\begin{subfigure}{.3\textwidth}
		\includegraphics[width=\textwidth]{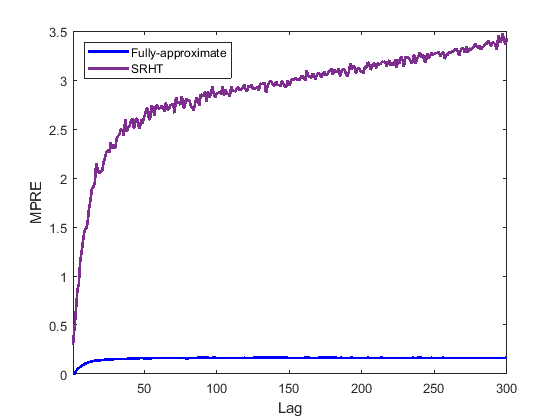}
		\caption{$\mathtt{AR(200)}$ }
		\label{fig:MaxRelErrAR(200)}
	\end{subfigure}
	\begin{subfigure}{.3\textwidth}
		\includegraphics[width=\textwidth]{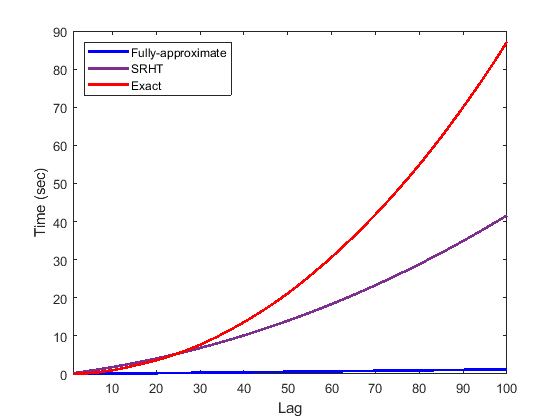}
		\caption{$\mathtt{AR(20)}$ }
		\label{fig:LSTimeAR(20)}
	\end{subfigure}
	\begin{subfigure}{.3\textwidth}
		\includegraphics[width=\textwidth]{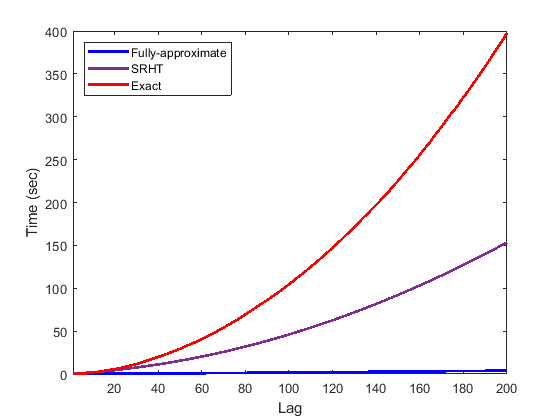}
		\caption{$\mathtt{AR(100)}$}
		\label{fig:LSTimeAR(100)}
	\end{subfigure}
	\begin{subfigure}{.3\textwidth}
		\includegraphics[width=\textwidth]{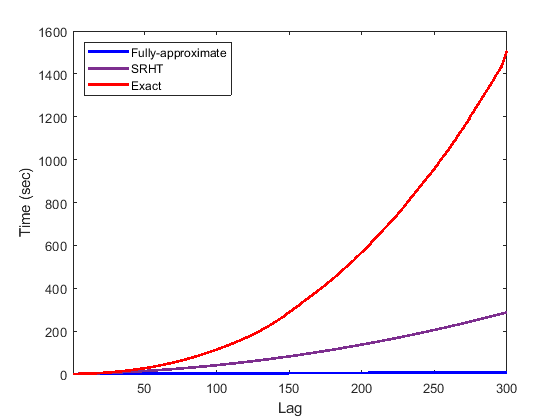}
		\caption{$\mathtt{AR(200)}$ }
		\label{fig:LSTimeAR(200)}
	\end{subfigure}	
	\caption{Figures (a), (b) and (c) correspond to $\mathtt{AR(20)}$, $\mathtt{AR(100)}$, and $\mathtt{AR(200)}$ using  synthetic data, respectively, and display the \texttt{MPRE} \cref{equ:MaxRelErrors} versus the lag values $h$ for fully-approximate and the SRHT method. Similarly, Figures (d), (e), and (f) represent the computational time spent, in seconds, to compute the fully-approximate leverage scores (in blue), the SRHT approximation (in magenta), and the exact leverage scores (in red) on $\mathtt{AR(20)}$, $\mathtt{AR(100)}$, and $\mathtt{AR(200)}$ using synthetic data, respectively.}
	\label{fig:LS_rel_err}
\end{figure}

\cref{fig:MaxRelErrAR(20),fig:MaxRelErrAR(100),fig:MaxRelErrAR(200)} suggest that the upper bound provided in \cref{thm:main} might be improved by replacing $p-1$ with an appropriate scaled function of $\log(p)$. This observation is numerically investigated in \cref{fig:Tthe4_ver}. In this figure (which in logarithmic scale), the \texttt{MPRE} \cref{equ:MaxRelErrors} (in blue) is compared with the right hand side (RHS) of \cref{thm:main} (in red) as well as the RHS of \cref{thm:main} with $p-1$ replaced with a scaled $\log(p)$ (in green). These results are in strong agreement with \cref{rem:Time Complexity,rem:Conj}. Indeed, improving the dependence of the RHS of \cref{thm:main} on $ p $ is an interesting problem, which we intend to address in future works.
\begin{figure}[h!]
	\centering
	\begin{subfigure}{.3\textwidth}
		\includegraphics[width=\textwidth]{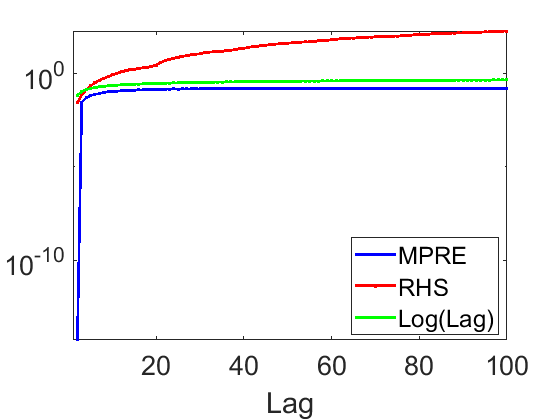}
		\caption{$\mathtt{AR(20)}$}
	\end{subfigure}
	\begin{subfigure}{.3\textwidth}
		\includegraphics[width=\textwidth]{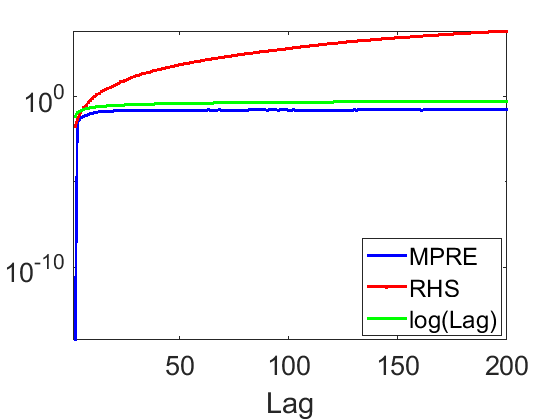}
		\caption{$\mathtt{AR(100)}$}
	\end{subfigure}
	\begin{subfigure}{.3\textwidth}
		\includegraphics[width=\textwidth]{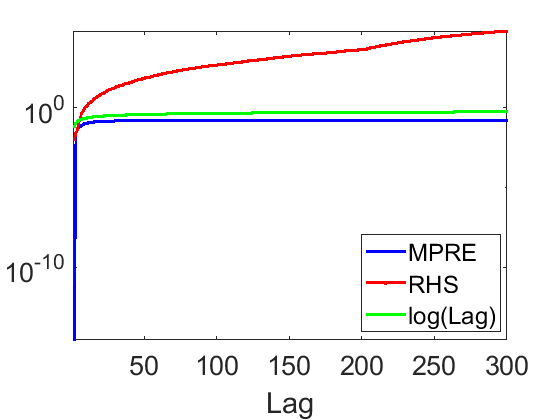}
		\caption{$\mathtt{AR(200)}$}
	\end{subfigure}
	\caption{Figures (a), (b) and (c) correspond to $\mathtt{AR(20)}$, $\mathtt{AR(100)}$, and $\mathtt{AR(200)}$ with synthetic data, respectively. Here, we display the \texttt{MPRE} \cref{equ:MaxRelErrors} (in blue), the RHS of \cref{thm:main} (in red) and RHS of \cref{thm:main} with $p-1$ replaced with a scaled $\log(p)$ (in green).}
	\label{fig:Tthe4_ver}
\end{figure}

\cref{fig:impact_n_s} exhibits the impact of the data size $n$ and the sample size $s$ on \texttt{MPRE} for the $\mathtt{AR(100)}$ synthetic data. More precisely, this figure demonstrates \texttt{MPRE} for values of $n \in \{500K, 1M, 2M\}$ (where, $K$ and $M$ stand for ``thousand'' and ``million'', respectively) and $s \in \{0.001n, 0.01n,$ $0.1n\}$. Clearly, for each fixed value of $n$, by increasing $s$, \texttt{MPRE} decreases. Furthermore, for each fixed ratio of $s/n$, by increasing $n$, s increases and accordingly \texttt{MPRE} decreases. It is clear that more data amounts to smaller approximation errors.
\begin{figure}[h!]
	\centering
	\begin{subfigure}{.3\textwidth}
		\includegraphics[width=\textwidth]{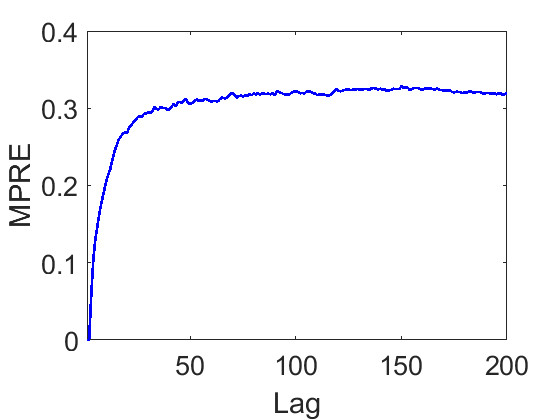}
		\caption{$n=500K, s=0.001 n$}
	\end{subfigure}
	\begin{subfigure}{.3\textwidth}
		\includegraphics[width=\textwidth]{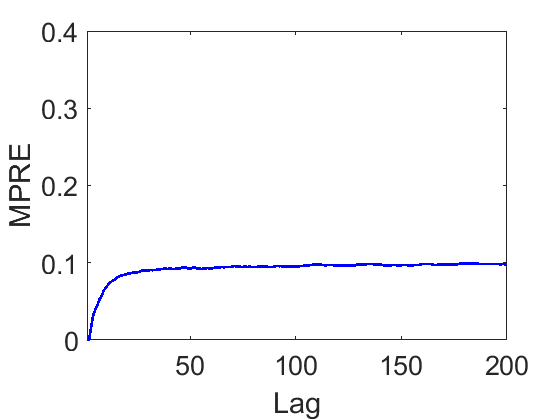}
		\caption{$n=500K, s=0.01 n$ }
	\end{subfigure}
	\begin{subfigure}{.3\textwidth}
		\includegraphics[width=\textwidth]{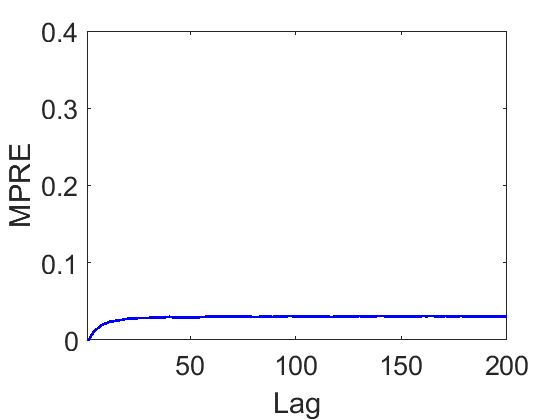}
		\caption{$n=500K, s=0.1 n$}
	\end{subfigure}
	\begin{subfigure}{.3\textwidth}
		\includegraphics[width=\textwidth]{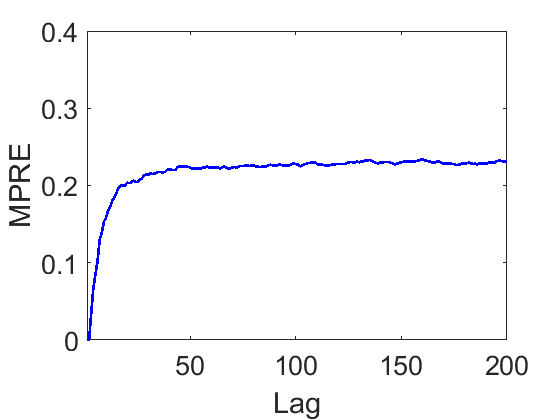}
		\caption{$n=1M, s=0.001 n$}
	\end{subfigure}
	\begin{subfigure}{.3\textwidth}
		\includegraphics[width=\textwidth]{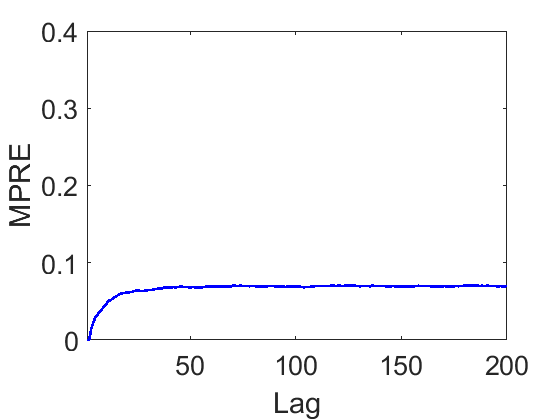}
		\caption{$n=1M, s=0.01 n$ }
	\end{subfigure}
	\begin{subfigure}{.3\textwidth}
		\includegraphics[width=\textwidth]{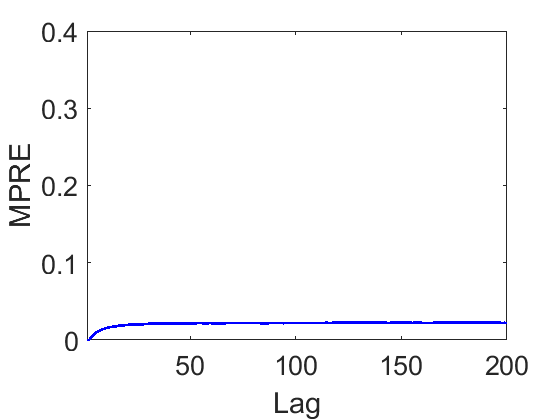}
		\caption{$n=1M, s=0.1 n$}
	\end{subfigure}
	\begin{subfigure}{.3\textwidth}
		\includegraphics[width=\textwidth]{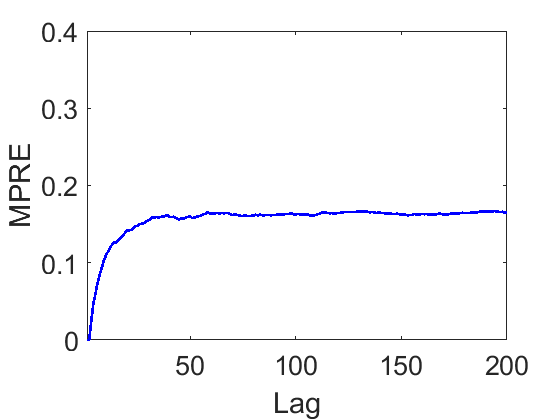}
		\caption{$n=2M, s=0.001 n$}
	\end{subfigure}
	\begin{subfigure}{.3\textwidth}
		\includegraphics[width=\textwidth]{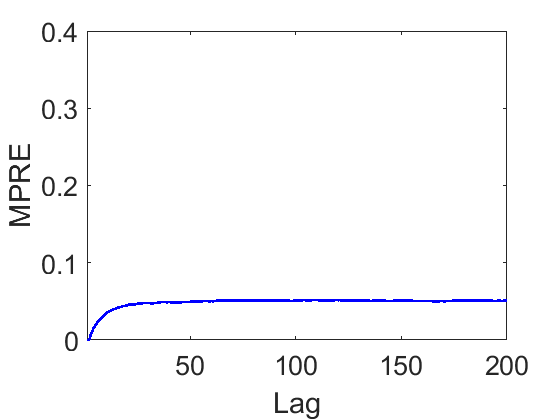}
		\caption{$n=2M, s=0.01 n$ }
	\end{subfigure}
	\begin{subfigure}{.3\textwidth}
		\includegraphics[width=\textwidth]{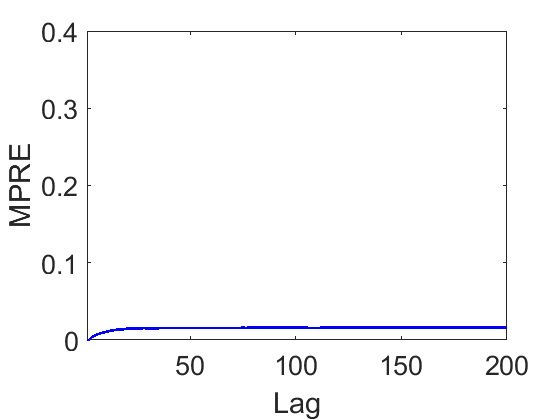}
		\caption{$n=2M, s=0.1 n$}
	\end{subfigure}
	\caption{The impact of the data size $n \in \{500K, 1M, 2M\}$ and the sample size $s \in \{0.001 n, 0.01 n, 0.1 n\}$ on \texttt{MPRE} for the $\mathtt{AR(100)}$ synthetic data.}
	\label{fig:impact_n_s}
\end{figure}

\subsection{PACF: Computational Time and Estimation Accuracy} \label{Sec:EmpRes-PACF}

In this section, using the same synthetic data as in \cref{Sec:EmpRes-LS}, we estimate PACF and fit an \texttt{AR} model. More precisely, for each dataset, PACF is estimated for a range of lag values $h$, once by solving the corresponding OLS problem with the full-data matrix (called, ``exact''), and another time by running the \texttt{LSAR} algorithm (\cref{alg:leverage}). 

The numerical results of these experiments for the three synthetic datasets are displayed in \cref{fig:PACF}. As explained in \cref{SecAR}, the most important application of a PACF plot is estimating the order $p$ by choosing the largest lag $h$ such that its corresponding PACF bar lies out of the $95\%$ zero-confidence boundary. It is readily seen that \cref{fig:PACF LSAR AR20,fig:PACF LSAR AR100,fig:PACF LSAR AR200} not only provide the correct estimate of the order $p$ for the generated synthetic data, but also are very close to the exact PACF plots in \cref{fig:PACF Exact AR20,fig:PACF Exact AR100,fig:PACF Exact AR200}. This is achieved all the while by merely sampling only $0.1\%$ of the rows of the data matrix (i.e., $s = 0.001n = 2000$). Subsequently, from \cref{fig:PACF time AR20,fig:PACF time AR100,fig:PACF time AR200},  one can observe a significant difference in the time required for computing PACF exactly as compared with obtaining a high-quality approximation using the \texttt{LSAR} algorithm.  
\begin{figure}[h!]
	\centering
	\begin{subfigure}{.3\textwidth}
		\includegraphics[width=\textwidth]{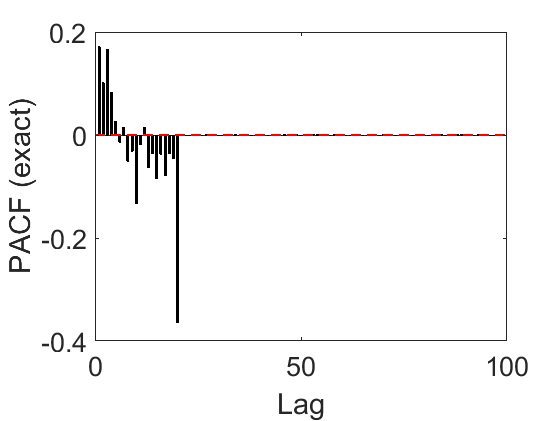}
		\caption{$\mathtt{AR(20)}$}
		\label{fig:PACF Exact AR20}
	\end{subfigure}
	\begin{subfigure}{.3\textwidth}
		\includegraphics[width=\textwidth]{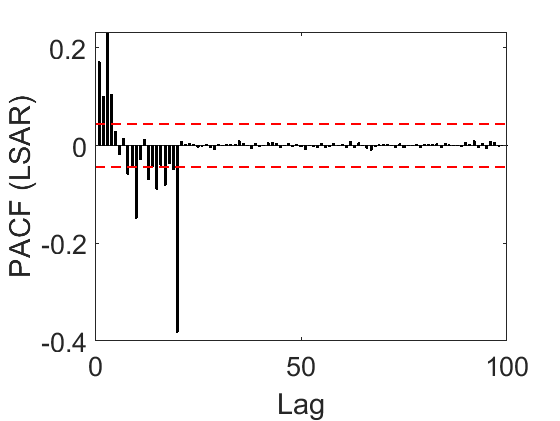}
		\caption{$\mathtt{AR(20)}$}
		\label{fig:PACF LSAR AR20}
	\end{subfigure}
	\begin{subfigure}{.3\textwidth}
		\includegraphics[width=\textwidth]{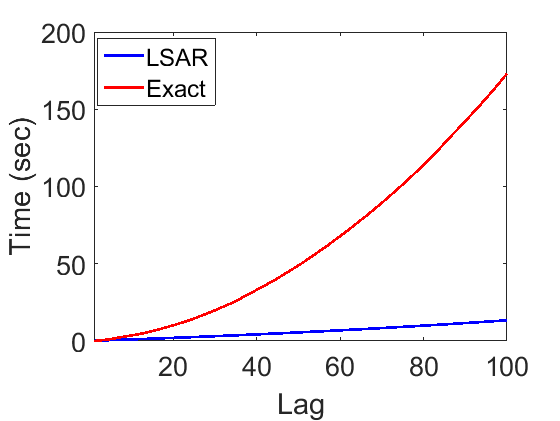}
		\caption{$\mathtt{AR(20)}$}
		\label{fig:PACF time AR20}
	\end{subfigure}
	\begin{subfigure}{.3\textwidth}
		\includegraphics[width=\textwidth]{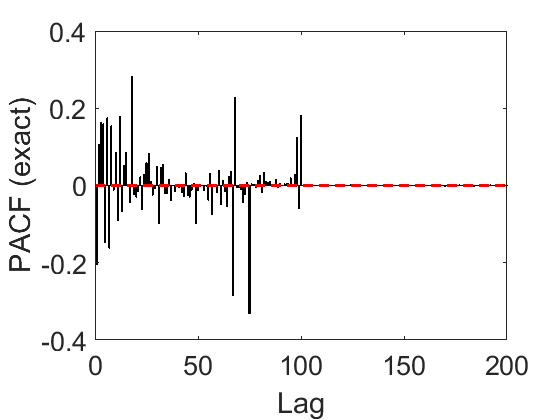}
		\caption{$\mathtt{AR(100)}$}
		\label{fig:PACF Exact AR100}
	\end{subfigure}
	\begin{subfigure}{.3\textwidth}
		\includegraphics[width=\textwidth]{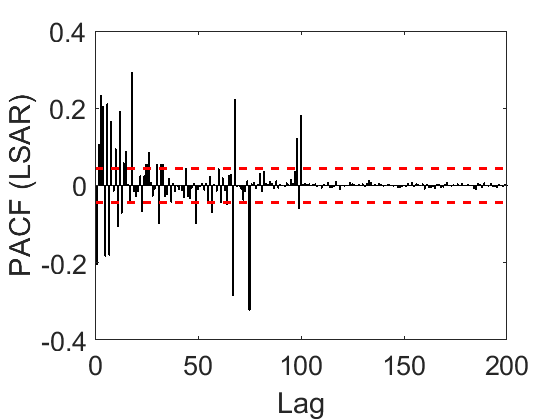}
		\caption{$\mathtt{AR(100)}$}
		\label{fig:PACF LSAR AR100}
	\end{subfigure}
	\begin{subfigure}{.3\textwidth}
		\includegraphics[width=\textwidth]{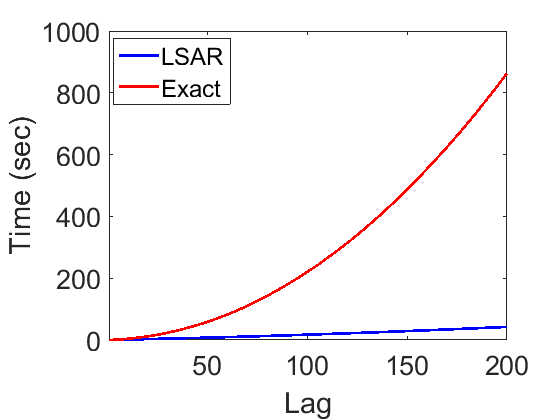}
		\caption{$\mathtt{AR(100)}$}
		\label{fig:PACF time AR100}
	\end{subfigure}
	\begin{subfigure}{.3\textwidth}
		\includegraphics[width=\textwidth]{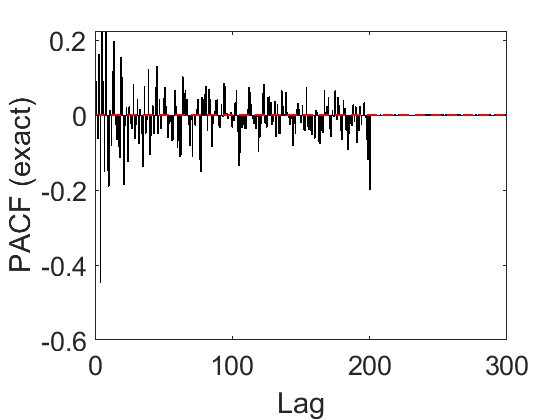}
		\caption{$\mathtt{AR(200)}$}
		\label{fig:PACF Exact AR200}
	\end{subfigure}
	\begin{subfigure}{.3\textwidth}
		\includegraphics[width=\textwidth]{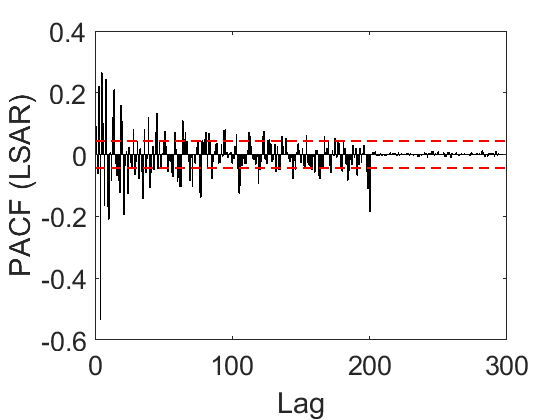}
		\caption{$\mathtt{AR(200)}$}
		\label{fig:PACF LSAR AR200}
	\end{subfigure}
	\begin{subfigure}{.3\textwidth}
		\includegraphics[width=\textwidth]{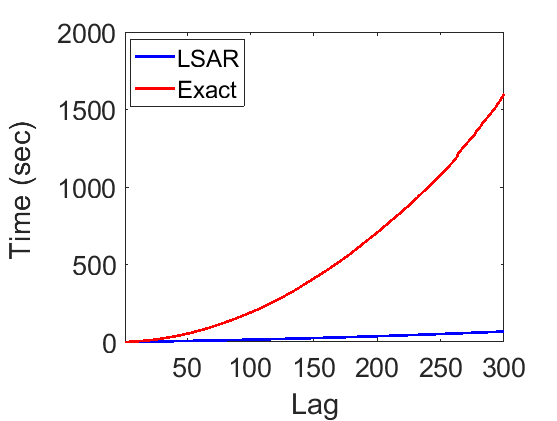}
		\caption{$\mathtt{AR(200)}$}
		\label{fig:PACF time AR200}
	\end{subfigure}
	\caption{Figures (a), (b) and (c) corresponding to the $\mathtt{AR(20)}$ synthetic data, display the exact PACF plot, the PACF plot generated by the \texttt{LSAR} algorithm, and the comparison between the computational time of (a) (in red) and (b) (in blue), respectively. Figures (d), (e) and (f) are similar for the $\mathtt{AR(100)}$ synthetic data; and Figures (g), (h) and (i) are similar for the $\mathtt{AR(200)}$ synthetic data.}
	\label{fig:PACF}
\end{figure}

\begin{remark} \label{rem:ls-time}
	Following \cref{rem:pacf}, finding PACF at each lag requires the solution to the corresponding OLS problem. Hence, to avoid duplication, the computational times of Steps 4-8 of the \texttt{LSAR} algorithm are excluded in \cref{fig:LS_rel_err}. Indeed, those computational times are considered in \cref{fig:PACF}. 
\end{remark}

To show the accuracy of maximum likelihood estimates generated by the \texttt{LSAR} algorithm, the estimates derived by the two scenarios of ``full-data matrix'' and ``reduced data matrix'' are relatively compared. For this purpose, following notation defined in \cref{Sec:Background,Sec:TheoreticalResults}, let $\bm{\phi}_{n,p}$ and $\hat{\bm{\phi}}^s_{n,p}$ denote the maximum likelihood estimates of parameters based on the full-data matrix (cf.\ \cref{eq:phi}) and reduced sampled data matrix with the sample size of $s$ (cf.\ \cref{eq:phi_hat}), respectively. Accordingly, we define the relative error of parameter estimates by
\begin{subequations}
	\begin{align}
	\label{eq:ratio_phi}
	\frac{||\hat{\bm{\phi}}^s_{n,p} - \bm{\phi}_{n,p}||}{||\bm{\phi}_{n,p}||}.
	\end{align}
	Analogously, let $\bm{r}_{n,p}$ and $\hat{\bm{r}}^s_{n,p}$ be the residuals of estimates based on the full-data matrix (cf.\ \cref{eq:res}) and reduced sampled data matrix with the sample size of $s$ (cf.\ \cref{eq:res_hat}), respectively. The ratio of two residual norms is given by
	\begin{align}
	\label{eq:ratio_res}
	\frac{\vnorm{\hat{\bm{r}}^s_{n,p}}}{\vnorm{\bm{r}_{n,p}}}.
	\end{align}
\end{subequations}

The two ratios \cref{eq:ratio_phi,eq:ratio_res} are calculated for a range of values of $s\in\{200, 300, \dots,$ $1000\}$ by computing the maximum likelihood estimates of the $\mathtt{AR(100)}$ synthetic data once with the full-data matrix and another time by running the \texttt{LSAR} algorithm. Also, the estimates are smoothed out by replicating the \texttt{LSAR} algorithm $1,000$ times and taking the average of all estimates. The outcome is displayed in \cref{fig:AR100_LS_SSE}. \cref{fig:AR100_relerr} displays the relative errors of parameter estimates \cref{eq:ratio_phi} versus the sample size $s$ and \cref{fig:AR100_relres} shows the ratio of residual norms \cref{eq:ratio_res} versus the sample size $s$. 
\begin{figure}[h!]
	\centering
	\begin{subfigure}{.3\textwidth}
		\includegraphics[width=\textwidth]{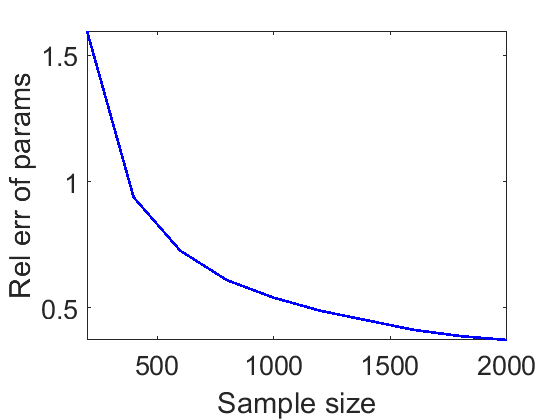}
		\caption{\cref{eq:ratio_phi} for the $\mathtt{AR(100)}$ data}
		\label{fig:AR100_relerr}
	\end{subfigure}\qquad
	\begin{subfigure}{.3\textwidth}
		\includegraphics[width=\textwidth]{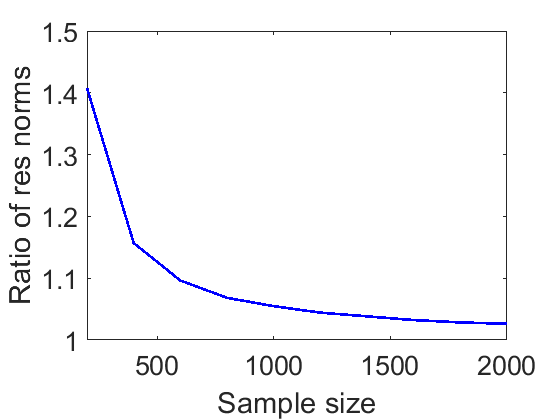}
		\caption{\cref{eq:ratio_res} for the $\mathtt{AR(100)}$ data}
		\label{fig:AR100_relres}
	\end{subfigure}
	\caption{Figures (a) and (b) display the relative error of parameter estimates \cref{eq:ratio_phi} and the ratio of residual norms \cref{eq:ratio_res} for the $\mathtt{AR(100)}$ synthetic data, respectively, both as a function of sample size $ s $.}
	\label{fig:AR100_LS_SSE}
\end{figure}

\subsection{Real-world Big Time Series: Gas Sensors Data} \label{Sec:EmpRes-RealData}

Huerta et al. \cite{Huerta2016OnlineHA} studied the accuracy of electronic nose measurements. They constructed a nose consisting of eight different metal-oxide sensors in addition to humidity and temperature sensors with a wireless communication channel to collect data. The nose monitored airflow for two years in a designated location, and data continuously collected with a rate of two observations per second. In this configuration, a standard energy band model for an $n-$type metal-oxide sensor was used to estimate the changes in air temperature and humidity. Based on their observations, humidity changes and correlated changes of humidity and temperature were the most significant statistical factors in variations of sensor conductivity. The model successfully used for gas discrimination with an R-squared close to $1$. 

The data is available in the UCI machine learning repository\footnote{\url{https://archive.ics.uci.edu/ml/datasets/Gas+sensors+for+home+activity+monitoring}, Accessed on 30 December 2019.}. In our experiment, we use the output of sensor number $8$ (column labeled R8 in the dataset) as real time series data with $n = 919,438$ observations. The original time series data is heteroscedastic. However, by taking the logarithm of the data and differencing in one lag, it becomes stationary and an $\mathtt{AR(16)}$ model seems to be a good fit to the transformed data. We run the \texttt{LSAR} algorithm with the initial input parameter $\bar{p} = 100$. Recalling that $\bar{p}$ ($p < \bar{p} \ll n$) is initially set large enough to estimate the order $p$ of the underlying \texttt{AR} model. Also, in all iterations of the \texttt{LSAR} algorithm, we set the sample size $s = 0.001 n$. For sake of fairness and completeness, all figures generated for synthetic data in \cref{Sec:EmpRes-LS,Sec:EmpRes-PACF}, are regenerated for the gas sensors data. 

\cref{fig:gas_LS_MPRE} shows (in logarithmic scale) the maximum pointwise relative error \cref{equ:MaxRelErrors} (in blue) along with the RHS of \cref{thm:main} (in red) as well as the RHS of \cref{thm:main} with $p-1$ replaced with a scaled $\log(p)$ (in green). The behavior of these three graphs are very similar to those ones on synthetic data discussed in \cref{Sec:EmpRes-LS}. Furthermore, \cref{fig:gas_LS_time} reveals analogous computational efficiency in finding the fully-approximate leverage scores comparing with the exact values for the gas sensors data.
\begin{figure}[h!]
	\centering
	\begin{subfigure}{.3\textwidth}
		\includegraphics[width=\textwidth]{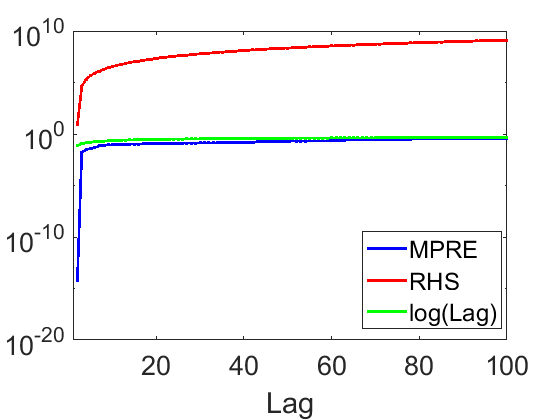}
		\caption{\texttt{MPRE}}
		\label{fig:gas_LS_MPRE}
	\end{subfigure}\qquad
	\begin{subfigure}{.3\textwidth}
		\includegraphics[width=\textwidth]{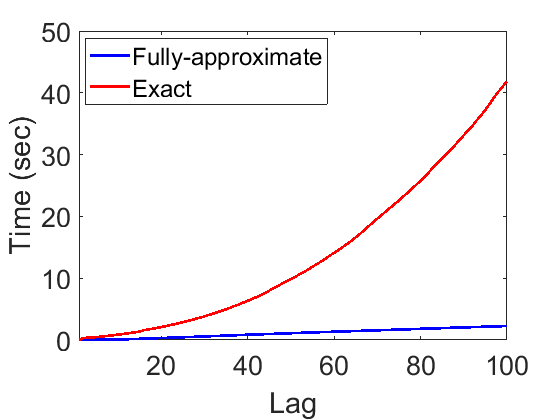}
		\caption{Computational time}
		\label{fig:gas_LS_time}
	\end{subfigure}
	
	\caption{Figure (a) displays the \texttt{MPRE} \cref{equ:MaxRelErrors} (in blue), the RHS of \cref{thm:main} (in red) and the RHS of \cref{thm:main} with $p-1$ replaced with a scaled $\log(p)$ (in green) for the real gas sensors data, Figure (b) shows the computational time spent, in seconds, to compute the fully-approximate (in blue) and exact (in red) leverage scores for the real gas sensors data.}
	\label{fig:gas_LS}
\end{figure}

\paragraph{Leverage score sampling versus uniform sampling.} In order to show the efficacy of the \texttt{LSAR} algorithm, we compare the performance of leverage score sampling with na\"{i}ve uniform sampling in estimating the order as well as parameters of an $\mathtt{AR}$ model for the gas sensor data. For the uniform sampling, we modify the \texttt{LSAR} algorithm slightly by removing Step 3 and replacing the uniform distribution $\hat{\pi}_{m,p}(i)=1/(m-p)$ for $i=1,\ldots,m-p$ in Step 4.

\cref{gas_PACF_Exact,gas_PACF_LSAR,gas_PACF_Uniform,gas_PACF_Time} demonstrate the PACF plot calculated exactly, the PACF plot approximated with the \texttt{LSAR} algorithm, the PACF plot approximated based on a uniform sampling, and a comparison between the computational times of these three PACF plots, respectively. Similar to \cref{Sec:EmpRes-PACF}, \cref{gas_PACF_Time} reveals that while \cref{gas_PACF_LSAR} can be generated much faster than \cref{gas_PACF_Exact}, they both suggest the same \texttt{AR} model for the gas sensors data. In addition, \cref{gas_PACF_Time,gas_PACF_Uniform} divulge that although the uniform sampling is slightly faster than the \texttt{LSAR} algorithm, the PACF plot generated by the former is very poor and far away from the exact plot given in \cref{gas_PACF_Exact}. While both \cref{gas_PACF_Exact,gas_PACF_LSAR} estimate an $\mathtt{AR(16)}$ model for the data, \cref{gas_PACF_Uniform} fails to make an appropriate estimate of the order. 
\begin{figure}[h!]
	\centering
	\begin{subfigure}{.4\textwidth}
		\includegraphics[width=\textwidth]{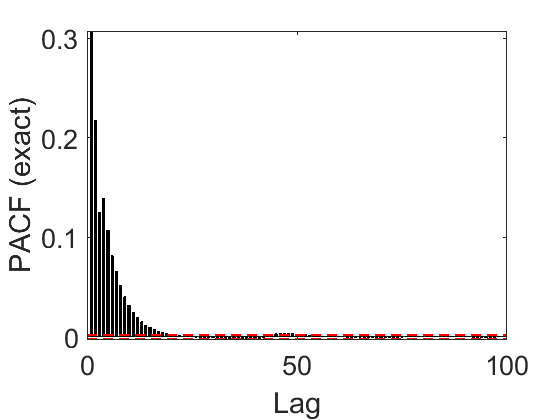}
		\caption{Exact}
		\label{gas_PACF_Exact}
	\end{subfigure}
	\begin{subfigure}{.4\textwidth}
		\includegraphics[width=\textwidth]{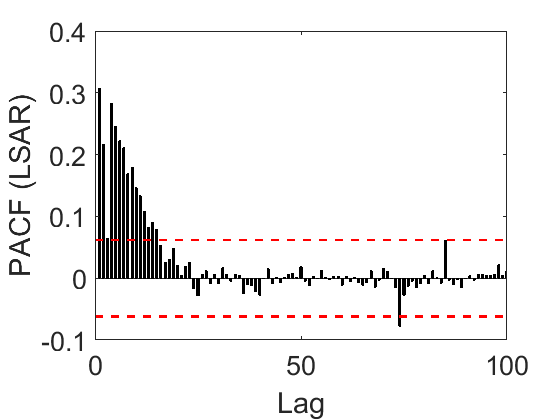}
		\caption{\texttt{LSAR}}
		\label{gas_PACF_LSAR}
	\end{subfigure}
	\begin{subfigure}{.4\textwidth}
		\includegraphics[width=\textwidth]{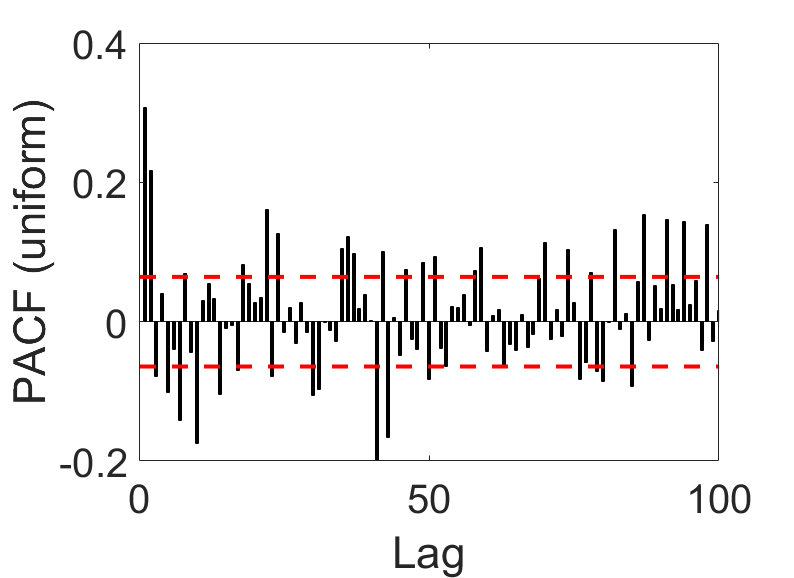}
		\caption{Uniform}
		\label{gas_PACF_Uniform}
	\end{subfigure}
	\begin{subfigure}{.4\textwidth}
		\includegraphics[width=\textwidth]{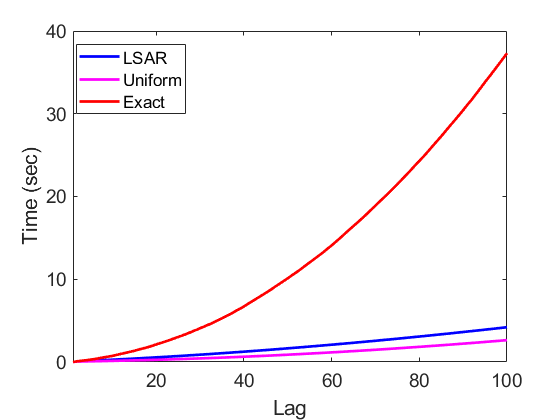}
		\caption{Computational time}
		\label{gas_PACF_Time}
	\end{subfigure}
	\caption{Figures (a), (b), (c) and (d) display the exact PACF plot, the PACF plot generated by the \texttt{LSAR} algorithm, the PACF plot approximated based on a uniform sampling, and the comparison between the computational time of (a) (in red), (b) (in blue), and (c) (in pink), respectively.} 
	\label{gas_PACF}
\end{figure}

Finally, we compare the performance of leverage score sampling with na\"{i}ve uniform sampling in estimating the parameters of an $\mathtt{AR(16)}$ model for the gas sensor data. \cref{fig:gas_LS_unif} compares the performance of these two sampling strategies for sample sizes chosen from $s \in \{200,300,\ldots,1000\}$. For each sampling scheme and a fixed sample size, the maximum likelihood estimates are smoothed out by replicating the \texttt{LSAR} algorithm $1,000$ times and taking the average of all estimates. Note that in all three \cref{fig:gas_LS_unif_relerr,fig:gas_LS_unif_relres,fig:gas_LS_unif_time}, the blue and red plots correspond with the leverage score and uniform sampling scheme, respectively.

\cref{fig:gas_LS_unif_relerr} displays the relative errors of parameter estimates \cref{eq:ratio_phi} and \cref{fig:gas_LS_unif_relres} shows the ratio of residual norms \cref{eq:ratio_res}, under the two sampling schemes. Both figures strongly suggest that the leverage score sampling scheme outperforms the uniform sampling scheme. Furthermore, while the output of the former shows stability and almost monotonic convergence, the latter exhibits oscillations and does not show any indication of convergence for such small sample sizes. This observation is consistent with the literature discussed in \cref{Sec:RandNLA}. Despite the fact uniform sampling can be performed almost for free, \cref{fig:gas_LS_unif_time} shows no significant difference between the computational time of both sampling scheme. 
\begin{figure}[h!]
	\centering
	\begin{subfigure}{.3\textwidth}
		\includegraphics[width=\textwidth]{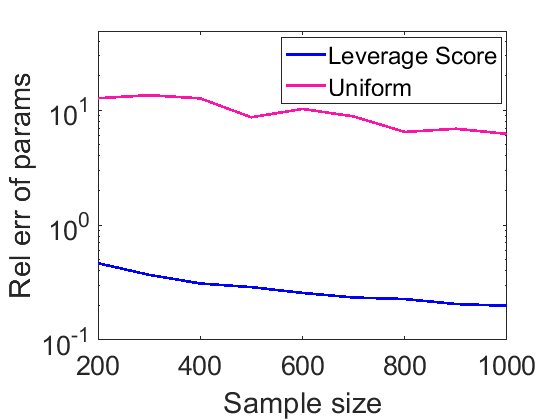}
		\caption{Relative error \cref{eq:ratio_phi}}
		\label{fig:gas_LS_unif_relerr}
	\end{subfigure}
	\begin{subfigure}{.3\textwidth}
		\includegraphics[width=\textwidth]{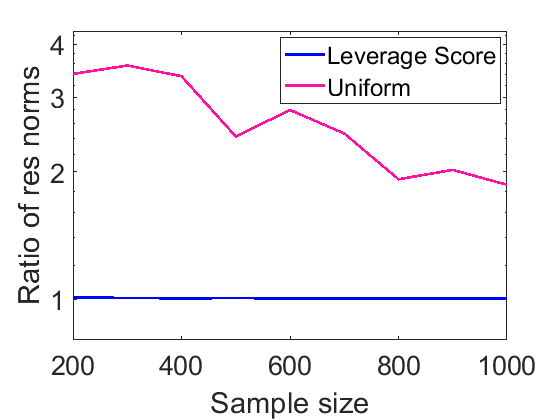}
		\caption{Ratio of residuals \cref{eq:ratio_res}}
		\label{fig:gas_LS_unif_relres}
	\end{subfigure}
	\begin{subfigure}{.3\textwidth}
		\includegraphics[width=\textwidth]{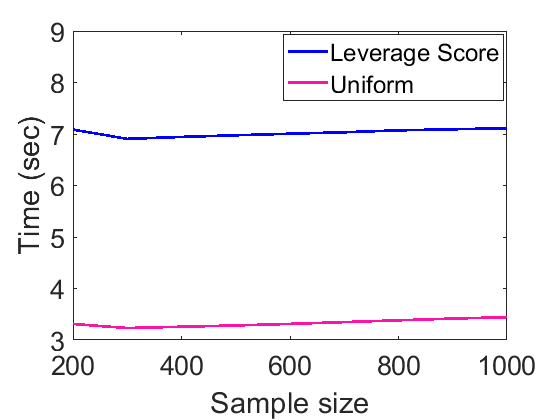}
		\caption{Computational time}
		\label{fig:gas_LS_unif_time}
	\end{subfigure}
	\caption{Figures (a), (b) and (c) display the relative error of parameter estimates \cref{eq:ratio_phi}, the ratio of residual norms \cref{eq:ratio_res}, and the computational time of two sampling schemes based on the leverage scores (blue) and uniform distribution (pink) for the gas sensors data, respectively.}
	\label{fig:gas_LS_unif}
\end{figure}

Finally, in our numerical examples, depending on the order of the \texttt{AR} model, the time difference between the exact method and the \texttt{LSAR} algorithm for model fitting vary between 75 to 1600 seconds. In many practical situations, one might need to fit hundreds of such models and make time-sensitive decisions based on the generated forecasts, before new data is provided. One such example is predicting several stock prices in a financial market for portfolio optimization, while the prices may be updated every few seconds. Another practical example is predicting the meteorology indices for several different purposes, with updates becoming available every few minutes. In these situations, saving a few seconds/minutes in forecasting can be crucial.  


\section{Conclusion} 
\label{Sec:Conclusion}

In this paper, we have developed a new approach to fit an \texttt{AR} model to big time series data. Motivated from the literature of RandNLA in dealing with large matrices, we construct a fast and efficient algorithm, called \texttt{LSAR}, to approximate the leverage scores corresponding to the data matrix of an \texttt{AR} model, to estimate the appropriate underlying order, and to find the conditional maximum likelihood estimates of its parameters.
Analytical error bounds are developed for such approximations and the worst case running time of the \texttt{LSAR} algorithm is derived. Empirical results on large-scale synthetic as well as big real time series data highly support the theoretical results and reveal the efficacy of this new approach. 

For future work, we are mainly interested in developing this approach for a more general \texttt{ARMA} model. However, unlike \texttt{AR}, the (conditional) log-likelihood function for \texttt{ARMA} is a complicated non-linear function such that (C)MLEs cannot be derived analytically. Thus, it may require to exploit not only RandNLA techniques, but also modern optimization algorithms in big data regime to develop an efficient leverage score sampling scheme for \texttt{ARMA} models.





\bibliographystyle{plain}
\bibliography{Biblio}

\begin{thebibliography}{10}

\bibitem{Abo}
M.~Abolghasemi, J.~Hurley, A.~Eshragh, and B.~Fahimnia.
\newblock Demand forecasting in the presence of systematic events: Cases in
  capturing sales promotions.
\newblock {\em International Journal of Production Economics}, 230:107892,
  2020.

\bibitem{Anderson1998}
C.W. Anderson, E.A. Stolz, and S.~Shamsunder.
\newblock Multivariate autoregressive models for classification of spontaneous
  electroencephalographic signals during mental tasks.
\newblock {\em IEEE Transactions on Biomedical Engineering}, 45(3):277--286,
  1998.

\bibitem{avron2019universal}
H.~Avron, M.~Kapralov, C.~Musco, C.~Musco, A.~Velingker, and A.~Zandieh.
\newblock {A universal sampling method for reconstructing signals with simple
  Fourier transforms}.
\newblock In {\em Proceedings of the 51st Annual ACM SIGACT Symposium on Theory
  of Computing}, pages 1051--1063. ACM, 2019.

\bibitem{avron2010blendenpik}
H.~Avron, P.~Maymounkov, and S.~Toledo.
\newblock Blendenpik: Supercharging {LAPACK}'s least-squares solver.
\newblock {\em SIAM Journal on Scientific Computing}, 32(3):1217--1236, 2010.

\bibitem{Blackwell2009TimeSeries}
P.J. Blackwell and R.A. Davis.
\newblock {\em Time Series: Theory and Methods}.
\newblock Springer Series in Statistics. Springer, 2009.

\bibitem{Box}
G.E.P. Box and G.M. Jenkins.
\newblock {\em Time Series Analysis, Forecasting and Control}.
\newblock Holden-Day, San Francisco, 1976.

\bibitem{Cha}
N.~Chakravarthy, A.~Spanias, L.D. Iasemidis, and K.~Tsakalis.
\newblock Autoregressive modeling and feature analysis of dna sequences.
\newblock {\em EURASIP Journal on Advances in Signal Processing}, 2004:13--28,
  2004.

\bibitem{clarkson2017low}
K.L Clarkson and D.P. Woodruff.
\newblock Low-rank approximation and regression in input sparsity time.
\newblock {\em Journal of the ACM (JACM)}, 63(6):54, 2017.

\bibitem{drineas2012fast}
P.~Drineas, M.~Magdon-Ismail, M.W. Mahoney, and D.P. Woodruff.
\newblock Fast approximation of matrix coherence and statistical leverage.
\newblock {\em Journal of Machine Learning Research}, 13(Dec):3475--3506, 2012.

\bibitem{drineas2011faster}
P.~Drineas, M.W. Mahoney, S.~Muthukrishnan, and T.~Sarl{\'o}s.
\newblock Faster least squares approximation.
\newblock {\em Numerische Mathematik}, 117(2):219--249, 2011.

\bibitem{Esh}
A.~Eshragh, B.~Ganim, and T.~Perkins.
\newblock The importance of environmental factors in forecasting australian
  power demand.
\newblock {\em arXiv preprint arXiv:1911.00817}, 2019.

\bibitem{estrin2019lslq}
R.~Estrin, D.~Orban, and M.A. Saunders.
\newblock {LSLQ}: An iterative method for linear least-squares with an error
  minimization property.
\newblock {\em SIAM Journal on Matrix Analysis and Applications},
  40(1):254--275, 2019.

\bibitem{fong2011lsmr}
D.C-L. Fong and M.~Saunders.
\newblock {LSMR}: An iterative algorithm for sparse least-squares problems.
\newblock {\em SIAM Journal on Scientific Computing}, 33(5):2950--2971, 2011.

\bibitem{golub1983matrix}
G.H. Golub and C.F. Van~Loan.
\newblock {\em Matrix Computations}.
\newblock Johns Hopkins paperback. Johns Hopkins University Press, 1983.

\bibitem{Hamilton1989}
J.D. Hamilton.
\newblock A new approach to the economic analysis of nonstationary time series
  and the business cycle.
\newblock {\em Econometrica}, 57(2):357--384, 1989.

\bibitem{Ham}
J.D. Hamilton.
\newblock {\em Time Series Analysis}.
\newblock Princeton University Press, New Jersey, 1994.

\bibitem{Huerta2016OnlineHA}
R.A. Huerta, T.S. Mosqueiro, J.~Fonollosa, N.F. Rulkov, and
  I.~Rodr{\'\i}guez-Luj{\'a}n.
\newblock Online humidity and temperature decorrelation of chemical sensors for
  continuous monitoring.
\newblock {\em Chemometrics and Intelligent Laboratory Systems},
  157(15):169--176, 2016.

\bibitem{ma2015statistical}
P.~Ma, M.W. Mahoney, and B.~Yu.
\newblock A statistical perspective on algorithmic leveraging.
\newblock {\em Journal of Machine Learning Research}, 16(1):861--911, 2015.

\bibitem{mahoney2011randomized}
M.W. Mahoney.
\newblock Randomized algorithms for matrices and data.
\newblock {\em Foundations and Trends{\textregistered} in Machine Learning},
  2011.

\bibitem{mahoney2016lecture}
M.W. Mahoney.
\newblock Lecture notes on randomized linear algebra.
\newblock {\em arXiv preprint arXiv:1608.04481}, 2016.

\bibitem{MM13_STOC}
X.~Meng and M.W. Mahoney.
\newblock Low-distortion subspace embeddings in input-sparsity time and
  applications to robust linear regression.
\newblock In {\em Proceedings of the 45th Annual ACM Symposium on Theory of
  Computing}, pages 91--100, 2013.

\bibitem{meng2014lsrn}
X.~Meng, M.A. Saunders, and M.W. Mahoney.
\newblock {LSRN}: A parallel iterative solver for strongly over-or
  underdetermined systems.
\newblock {\em SIAM Journal on Scientific Computing}, 36(2):C95--C118, 2014.

\bibitem{Messner2019}
J.W. Messner and P.~Pinson.
\newblock Online adaptive lasso estimation in vector autoregressive models for
  high dimensional wind power forecasting.
\newblock {\em International Journal of Forecasting}, 35(4):1485--1498, 2019.

\bibitem{NN13}
J.~Nelson and H.L. Nguyen.
\newblock {OSNAP}: Faster numerical linear algebra algorithms via sparser
  subspace embeddings.
\newblock In {\em Proceedings of the 54th Annual IEEE Symposium on Foundations
  of Computer Science}, pages 117--126, 2013.

\bibitem{paige1982lsqr}
C.C. Paige and M.A. Saunders.
\newblock {LSQR}: An algorithm for sparse linear equations and sparse least
  squares.
\newblock {\em ACM Transactions on Mathematical Software (TOMS)}, 8(1):43--71,
  1982.

\bibitem{raskutti2016statistical}
G.~Raskutti and M.W. Mahoney.
\newblock A statistical perspective on randomized sketching for ordinary
  least-squares.
\newblock {\em Journal of Machine Learning Research}, 17(1):7508--7538, 2016.

\bibitem{roszas}
F.~Roosta-Khorasani, G.J. Sz\'{e}kely, and U.M. Ascher.
\newblock Assessing stochastic algorithms for large scale nonlinear least
  squares problems using extremal probabilities of linear combinations of gamma
  random variables.
\newblock {\em SIAM/ASA Journal on Uncertainty Quantification}, 3(1):61--90,
  2015.

\bibitem{rousseeuw2011robust}
P.J. Rousseeuw and M.~Hubert.
\newblock Robust statistics for outlier detection.
\newblock {\em Wiley Interdisciplinary Reviews: Data Mining and Knowledge
  Discovery}, 1(1):73--79, 2011.

\bibitem{She}
X.~Shen and Q.~Lu.
\newblock Joint analysis of genetic and epigenetic data using a conditional
  autoregressive model.
\newblock {\em BMC Genetics}, 16(Suppl 1):51--54, 2018.

\bibitem{shi2019sublinear}
X.~Shi and D.P. Woodruff.
\newblock Sublinear time numerical linear algebra for structured matrices.
\newblock In {\em Proceedings of the AAAI}, 2019.

\bibitem{Shu}
R.H. Shumway and D.S. Stoffer.
\newblock {\em Time Series Analysis and Its Applications}.
\newblock Springer, London, 2017.

\bibitem{van2003superfast}
Marc Van~Barel, Georg Heinig, and Peter Kravanja.
\newblock {A superfast method for solving Toeplitz linear least squares
  problems}.
\newblock {\em Linear algebra and its applications}, 366:441--457, 2003.

\bibitem{wang2019information}
H.Y. Wang, M.~Yang, and J.~Stufken.
\newblock Information-based optimal subdata selection for big data linear
  regression.
\newblock {\em Journal of the American Statistical Association},
  114(525):393--405, 2018.

\bibitem{woodruff2014sketching}
D.P. Woodruff.
\newblock Sketching as a tool for numerical linear algebra.
\newblock {\em Foundations and Trends{\textregistered} in Theoretical Computer
  Science}, 2014.

\bibitem{xi2014superfast}
Yuanzhe Xi, Jianlin Xia, Stephen Cauley, and Venkataramanan Balakrishnan.
\newblock {Superfast and stable structured solvers for Toeplitz least squares
  via randomized sampling}.
\newblock {\em SIAM Journal on Matrix Analysis and Applications}, 35(1):44--72,
  2014.

\bibitem{yang2015implementing}
J.~Yang, X.~Meng, and M.W. Mahoney.
\newblock Implementing randomized matrix algorithms in parallel and distributed
  environments.
\newblock In {\em Proceedings of the IEEE}, pages 58--92, 2016.

\end{thebibliography}


\appendix
\section{Technical Lemmas and Proofs}
\label{Sec:Proofs}

\subsection{Proof of \cref{thm:TheRec4LevScr}}

We first present \cref{LemInvBlockMat} which is used in the proof of \cref{thm:TheRec4LevScr}. 
\begin{lemma}[Matrix Inversion Lemma \cite{golub1983matrix}] \label{LemInvBlockMat}
	Consider the $2\times 2$ block matrix
	\begin{align*}
	\medskip M & = \begin{pmatrix}
	\medskip c & \bm{b}^{\transpose} \\
	\medskip \bm{b} & \AA 
	\end{pmatrix},
	\end{align*}
	where $\AA$, $\bm{b}$, and $c$ are an $m\times m$ matrix, an $m\times 1$ vector and a scalar, respectively. If $\AA$ is invariable, the inverse of matrix $M$ exists an can be calculated as follows
	\begin{align*}
	\medskip M^{-1} & = \frac{1}{k}\begin{pmatrix}
	\medskip 1 & -\bm{b}^{\transpose} \AA^{-1} \\
	\medskip -\AA^{-1} \bm{b} & k \AA^{-1} + \AA^{-1} \bm{b} \bm{b}^{\transpose} \AA^{-1}
	\end{pmatrix},
	\end{align*}
	where $k = c - \bm{b}^{\transpose} \AA^{-1} \bm{b}$.
\end{lemma}

\subsubsection*{Proof of \cref{thm:TheRec4LevScr}}
\begin{proof}
	For $ p=1 $, computing the leverage score trivially boils down to normalizing the data vector.
	For $ p \geq 2 $, the data matrix is given by
	\begin{align*}
	\medskip \XX_{n,p} & = \begin{pmatrix}
	\bm{y}_{n-1,p-1} & \XX_{n-1,p-1}
	\end{pmatrix}.
	\end{align*}
	So, we have
	\begin{align*}
	\medskip \XX_{n,p}^{\transpose} \XX_{n,p} & = \begin{pmatrix}
	\medskip \bm{y}_{n-1,p-1}^{\transpose} \bm{y}_{n-1,p-1} & \bm{y}_{n-1,p-1}^{\transpose} \XX_{n-1,p-1} \\
	\medskip \XX_{n-1,p-1}^{\transpose} \bm{y}_{n-1,p-1} & \XX_{n-1,p-1}^{\transpose} \XX_{n-1,p-1}
	\end{pmatrix}.
	\end{align*}
	For sake of simplicity, let us define
	\begin{align*}
	\medskip \WW_{n,p} \defeq \XX_{n,p}^{\transpose} \XX_{n,p}.
	\end{align*}
	Following \cref{LemInvBlockMat}, the inverse of matrix $\WW_{n,p}$ is given by
	\begin{align*}
	\WW^{-1}_{n,p}	\medskip & = \frac{1}{u_{n,p}} \begin{pmatrix}
	\medskip 1 & -\bm{\phi}_{n-1,p-1}^{\transpose} \\
	\medskip -\bm{\phi}_{n-1,p-1} & u_{n,p} \WW_{n-1,p-1}^{-1} + \bm{\phi}_{n-1,p-1} \bm{\phi}_{n-1,p-1}^{\transpose}
	\end{pmatrix},
	\end{align*}
	where 
	\begin{align*}
	\medskip u_{n,p} & \defeq \bm{y}_{n-1,p-1}^{\transpose} \bm{y}_{n-1,p-1} - \bm{y}_{n-1,p-1}^{\transpose} \XX_{n-1,p-1} \WW_{n-1,p-1}^{-1} \XX_{n-1,p-1}^{\transpose} \bm{y}_{n-1,p-1} \\
	\medskip & = \bm{y}_{n-1,p-1}^{\transpose} \bm{y}_{n-1,p-1} - \bm{y}_{n-1,p-1}^{\transpose} \XX_{n-1,p-1} \bm{\phi}_{n-1,p-1} .
	\end{align*}
	It is readily seen that 
	\begin{align*}
	\medskip u_{n,p} & \defeq \bm{y}_{n-1,p-1}^{\transpose} \bm{y}_{n-1,p-1} - 2\bm{y}_{n-1,p-1}^{\transpose} \XX_{n-1,p-1} \bm{\phi}_{n-1,p-1} + \bm{y}_{n-1,p-1}^{\transpose} \XX_{n-1,p-1} \bm{\phi}_{n-1,p-1} \\
	\medskip  & = \bm{y}_{n-1,p-1}^{\transpose} \bm{y}_{n-1,p-1} - 2\bm{y}_{n-1,p-1}^{\transpose} \XX_{n-1,p-1} \bm{\phi}_{n-1,p-1} \\ 
	\medskip & \hspace*{0.5cm} + \bm{y}_{n-1,p-1}^{\transpose} \XX_{n-1,p-1} \WW_{n-1,p-1}^{-1} \WW_{n-1,p-1}\bm{\phi}_{n-1,p-1} \\ 
	\medskip  & = \bm{y}_{n-1,p-1}^{\transpose} \bm{y}_{n-1,p-1} - 2\bm{y}_{n-1,p-1}^{\transpose} \XX_{n-1,p-1} \bm{\phi}_{n-1,p-1} + \bm{\phi}_{n-1,p-1}^{\transpose} \WW_{n-1,p-1}\bm{\phi}_{n-1,p-1} \\
	\medskip  & = \bm{y}_{n-1,p-1}^{\transpose} \bm{y}_{n-1,p-1} - 2\bm{y}_{n-1,p-1}^{\transpose} \XX_{n-1,p-1} \bm{\phi}_{n-1,p-1} \\
	\medskip & \hspace*{0.5cm} + \bm{\phi}_{n-1,p-1}^{\transpose} \XX_{n-1,p-1}^{\transpose} \XX_{n-1,p-1} \bm{\phi}_{n-1,p-1} \\	   
	\medskip  & = \vnorm{\bm{y}_{n-1,p-1} - \XX_{n-1,p-1} \bm{\phi}_{n-1,p-1}}^2 \\
	\medskip  & = \vnorm{\res_{n-1,p-1}}^2. 
	\end{align*}
	
	\noindent The $i^{th}$ leverage score is given by
	\begin{align*}
	\medskip \ell_{n,p}(i) & = \XX_{n,p}^{\transpose}(i,:) \WW_{n,p}^{-1} \XX_{n,p}(i,:) \\
	\medskip  & = \begin{bmatrix}
	y_{i+p-1} & \XX_{n-1,p-1}^{\transpose}(i,:) 
	\end{bmatrix} \WW_{n,p}^{-1} \begin{bmatrix}
	y_{i+p-1} \\ \XX_{n-1,p-1}(i,:) 
	\end{bmatrix} \\
	\medskip & = \frac{1}{\vnorm{\res_{n-1,p-1}}^2}
	\left[y_{i+p-1} - \XX_{n-1,p-1}^{\transpose}(i,:)\bm{\phi}_{n-1,p-1} \right. \\
	\medskip & \hspace*{0.5cm} \left. - y_{i+p-1}\bm{\phi}_{n-1,p-1}^{\transpose} + \XX_{n-1,p-1}^{\transpose}(i,:) (\vnorm{\res_{n-1,p-1}}^2 \WW_{n-1,p-1}^{-1} + \bm{\phi}_{n-1,p-1} \bm{\phi}_{n-1,p-1}^{\transpose})\right] \\
	\medskip & \hspace*{0.5cm} \times \begin{bmatrix}
	y_{i+p-1}  \\ \XX_{n-1,p-1}(i,:)
	\end{bmatrix} \\
	\medskip & = \frac{1}{\vnorm{\res_{n-1,p-1}}^2} \left(
	y_{i+p-1}^2 - \XX_{n-1,p-1}^{\transpose}(i,:)\bm{\phi}_{n-1,p-1}y_{i+p-1} - y_{i+p-1}\bm{\phi}_{n-1,p-1}^{\transpose}\XX_{n-1,p-1}(i,:) \right. \\
	\medskip & \hspace*{0.5cm}  \left. + \XX_{n-1,p-1}^{\transpose}(i,:) (\vnorm{\res_{n-1,p-1}}^2 \WW_{n-1,p-1}^{-1} + \bm{\phi}_{n-1,p-1} \bm{\phi}_{n-1,p-1}^{\transpose})\XX_{n-1,p-1}(i,:)\right) \\ 
	\medskip & = \XX_{n-1,p-1}^{\transpose}(i,:) \WW_{n-1,p-1}^{-1} \XX_{n-1,p-1}(i,:) \\
	\medskip & \hspace*{0.5cm} + \frac{1}{\vnorm{\res_{n-1,p-1}}^2}\left(y_{i+p-1}^2 - 2y_{i+p-1}\XX_{n-1,p-1}^{\transpose}(i,:)\bm{\phi}_{n-1,p-1} + (\XX_{n-1,p-1}^{\transpose}(i,:) \bm{\phi}_{n-1,p-1})^2
	\right) \\
	\medskip & = \ell_{n-1,p-1}(i) + \frac{1}{\vnorm{\res_{n-1,p-1}}^2} \vnorm{y_{i+p-1} - \XX_{n-1,p-1}^{\transpose}(i,:)\bm{\phi}_{n-1,p-1}}^2 \\
	\medskip  & = \ell_{n-1,p-1}(i) + \frac{\res_{n-1,p-1}^2(i)}{\vnorm{\res_{n-1,p-1}}^2}.
	\end{align*}
\end{proof}

\subsection{\cref{thm:quasi} and Its Proof}

\begin{theorem}[Relative Errors for Quasi-approximate Leverage Scores]
	\label{thm:quasi}
	For the quasi-approximate leverage scores, we have with probability at least $ 1-\delta $,
	\begin{align*}
	\frac{|\ell_{n,p}(i)-\tilde{\ell}_{n,p}(i)|}{\ell_{n,p}(i)} & \leq \left( 1 + 3 \eta_{n-1,p-1} \kappa^2(\XX_{n,p})\right) \sqrt{\varepsilon},\quad \mbox{for}\ i = 1,\ldots,n-p,
	\end{align*}
	recalling that $\eta_{n,p}, \kappa(\XX_{n,p})$, and $\varepsilon$ are as in \cref{thm:drma}. 
\end{theorem}

In order to prove \cref{thm:quasi}, we first introduce the following lemmas and corollary.
\begin{lemma}\label{LemNormOpt}
	The leverage scores of an $\mathtt{AR(p)}$ model for $p \geq 1$, are given by
	\begin{align*}
	\ell_{n,p}(i) = \min_{\zz \in \mathbb{R}^{n-p}} \left\{\vnorm{\zz}^2 \mid \XX_{n,p}^{\transpose} \zz = \XX_{n,p}(i,:) \right\}, \quad \mbox{for}\ i=1,\ldots,n-p,
	\end{align*}
	where $ \XX_{n,p} $ is the data matrix of the $ \mathtt{AR}(p) $ model defined in \cref{eq:Xnp}.
\end{lemma}
\begin{proof}
	We prove this lemma by using Lagrangian multipliers. Define the function
	\begin{align*}
	\medskip h(\zz,\bm{\lambda}) & \defeq \frac{1}{2}\zz^{\transpose}\zz-\bm{\lambda}^{\transpose}(\XX_{n,p}^{\transpose}\zz-\XX_{n,p}(i,:)).
	\end{align*}
	By taking the first derivative with respect to the vector $\zz$ and setting equal to zero, we have, 
	\begin{align*}
	\medskip \frac{\partial h(\zz,\bm{\lambda})}{\partial \zz} & = \zz-\XX_{n,p}\bm{\lambda}\ =\ 0\ \Rightarrow\ \zz^{\star}\ =\ \XX_{n,p}\bm{\lambda}^{\star}.
	\end{align*}
	Now, by multiplying both sides by $\XX_{n,p}^{\transpose}$, we obtain,
	\begin{align*}
	\medskip \XX_{n,p}^{\transpose} \zz^{\star} & = \XX_{n,p}^{\transpose} \XX_{n,p}\bm{\lambda}^{\star},
	\end{align*}
	simplified to
	\begin{align*}
	\medskip \XX_{n,p}(i,:) & = \XX_{n,p}^{\transpose} \XX_{n,p}\bm{\lambda}^{\star}.
	\end{align*}
	This implies that, 
	\begin{align*}
	\medskip \bm{\lambda}^{\star} & = (\XX_{n,p}^{\transpose} \XX_{n,p})^{-1}\XX_{n,p}(i,:).
	\end{align*}
	Thus,
	\begin{align*}
	\medskip \zz^{\star} & = \XX_{n,p}(\XX_{n,p}^{\transpose} \XX_{n,p})^{-1}\XX_{n,p}(i,:).
	\end{align*}
	The square of the norm of $\zz^{\star}$ is equal to
	\begin{align*}
	\medskip \vnorm{\zz^{\star}}^2 & = \left(\XX_{n,p}^{\transpose}(i,:) (\XX_{n,p}^{\transpose} \XX_{n,p})^{-1}\XX_{n,p}^{\transpose}\right)\left(\XX_{n,p}(\XX_{n,p}^{\transpose} \XX_{n,p})^{-1}\XX_{n,p}(i,:)\right) \\
	\medskip  & = \XX_{n,p}^{\transpose}(i,:) (\XX_{n,p}^{\transpose} \XX_{n,p})^{-1}\XX_{n,p}(i,:) \\
	\medskip  & = \ell_{n,p}(i).
	\end{align*}	
\end{proof}

\begin{lemma}\label{LemUpBnd4NormX}
	For an $\mathtt{AR(p)}$ model with $p \geq 1$, we have 
	\begin{align*}
	\medskip \vnorm{\XX_{n,p}(i,:)} & \leq \vnorm{\XX_{n,p}}\sqrt{\ell_{n,p}(i)}, \quad \mbox{for}\ i=1,\ldots,n-p,
	\end{align*} 	
	where $\XX_{n,p} $ and $\ell_{n,p}(i)$ are defined, respectively, in \cref{eq:Xnp,def:notation}.
\end{lemma}
\begin{proof}
	From \cref{LemNormOpt} we have,
	\begin{align*}
	\medskip \vnorm{\XX_{n,p}(i,:)} & = \vnorm{\XX_{n,p}^{\transpose}\zz^{\star}} \\
	\medskip  & \leq \vnorm{\XX_{n,p}^{\transpose}}\vnorm{\zz^{\star}} \\
	\medskip  & = \vnorm{\XX_{n,p}} \sqrt{\ell_{n,p}(i)}.
	\end{align*}
\end{proof}

\begin{lemma}\label{LemBnd4DiffError}
	For an $\mathtt{AR(p)}$ model with $p \geq 1$, we have
	\begin{align*}
	\medskip |\res_{n,p}(i)-\tilde{\res}_{n,p}(i)| & \leq \sqrt{\varepsilon} \eta_{n,p} \vnorm{\bm{\phi}_{n,p}} \vnorm{\XX_{n,p}} \sqrt{\ell_{n,p}(i)},\quad \mbox{for}\ i=1,\ldots,n-p,	 	
	\end{align*}
	where $ \res_{n,p}, \tilde{\res}_{n,p}, \eta_{n,p} , \bm{\phi}_{n,p}, \XX_{n,p}$, and $\ell_{n,p}(i) $ are defined respectively in \cref{eq:res}, \cref{eq:res_tilde}, \cref{eq:eta}, \cref{eq:phi}, \cref{eq:Xnp}, and  \cref{def:notation} and $ \varepsilon $ is the error in \cref{EquBoundSSE}.
\end{lemma}
\begin{proof}
	From \cref{EquRelDiffError} and the definition of $l_2$ norm, we have
	\begin{align*}
	\medskip \dotprod{\frac{\XX_{n,p}(i,:)}{\vnorm{\XX_{n,p}(i,:)}},(\bm{\phi}_{n,p}-\tilde{\bm{\phi}}_{n,p})} & \leq \vnorm{\bm{\phi}_{n,p}-\tilde{\bm{\phi}}_{n,p}} \\
	\medskip  & \leq \sqrt{\varepsilon} \eta_{n,p} \vnorm{\bm{\phi}_{n,p}}.		
	\end{align*}
	So, we have
	\begin{align*}
	\medskip \XX_{n,p}^{\transpose}(i,:)\bm{\phi}_{n,p} - \XX_{n,p}^{\transpose}(i,:)\tilde{\bm{\phi}}_{n,p} & \leq \sqrt{\varepsilon} \eta_{n,p}  \vnorm{\bm{\phi}_{n,p}} \vnorm{\XX_{n,p}(i,:)}.		
	\end{align*}
	Now by adding and subtracting $y_{i+p}$ on the left hand side, we yield
	\begin{align*}
	\medskip \left(y_{i+p} - \XX_{n,p}^\transpose(i,:)\tilde{\bm{\phi}}_{n,p}\right) - \left(y_{i+p} - \XX_{n,p}^\transpose(i,:)\bm{\phi}_{n,p}\right) & \leq \sqrt{\varepsilon} \eta_{n,p} \vnorm{\bm{\phi}_{n,p}} \vnorm{\XX_{n,p}(i,:)}, 		
	\end{align*}
	implying that,
	\begin{align*}
	\medskip \tilde{\res}_{n,p}(i) - \res_{n,p}(i) & \leq \sqrt{\varepsilon} \eta_{n,p} \vnorm{\bm{\phi}_{n,p}} \vnorm{\XX_{n,p}(i,:)}.		
	\end{align*}
	As analogously we can construct a similar inequality for $\res_{n,p}(i) - \tilde{\res}_{n,p}(i)$, we have that
	\begin{align*}
	\medskip |\res_{n,p}(i)-\tilde{\res}_{n,p}(i)| & \leq \sqrt{\varepsilon} \eta_{n,p} \vnorm{\bm{\phi}_{n,p}} \vnorm{\XX_{n,p}(i,:)}.		
	\end{align*}
	Now, by using \cref{LemUpBnd4NormX}, we obtain
	\begin{align*}
	\medskip |\res_{n,p}(i)-\tilde{\res}_{n,p}(i)| & \leq \sqrt{\varepsilon} \eta_{n,p} \vnorm{\bm{\phi}_{n,p}} \vnorm{\XX_{n,p}} \sqrt{\ell_{n,p}(i)}.
	\end{align*}
\end{proof}

\begin{lemma}\label{LemBnd4Error}
	Let $\{y_1,\ldots,y_n\}$ be a time series data. For $i=1,\ldots,n-p$, we have
	\begin{subequations}
		\begin{align}
		\medskip \label{EquBnd4Error} |\res_{n-1,p-1}(i)| & \leq \sqrt{\vnorm{\bm{\phi}_{n-1,p-1}}^2 + 1} \vnorm{\XX_{n,p}} \sqrt{\ell_{n,p}(i)}, \\	 	
		\medskip \label{EquBnd4ErrorHat} |\tilde{\res}_{n-1,p-1}(i)| & \leq \sqrt{\vnorm{\tilde{\bm{\phi}}_{n-1,p-1}}^2 + 1} \vnorm{\XX_{n,p}} \sqrt{\ell_{n,p}(i)},	 	
		\end{align}	
	\end{subequations}
	where $ \res_{n,p}, \tilde{\res}_{n,p}, \bm{\phi}_{n,p}, \tilde{\bm{\phi}}_{n,p}, \XX_{n,p}$, and $\ell_{n,p}(i) $ are defined respectively in \cref{eq:res}, \cref{eq:res_tilde}, \cref{eq:phi}, \cref{eq:phi_tilde}, \cref{eq:Xnp}, and  \cref{def:notation}.
	
\end{lemma}
\begin{proof}
	The left hand side of \cref{EquBnd4Error} can be written as below: 
	\begin{align*}
	\medskip |\res_{n-1,p-1}(i)| & = |y_{i+p-1} - \XX_{n-1,p-1}^{\transpose}(i,:)\bm{\phi}_{n-1,p-1}| \\
	\medskip  & = |\begin{bmatrix}
	y_{i+p-1} & \XX_{n-1,p-1}^{\transpose}(i,:) \end{bmatrix} \begin{bmatrix} 
	1 & -\bm{\phi_{n-1,p-1}}^{\transpose} \end{bmatrix}^{\transpose}| \\
	\medskip  & = |\XX_{n,p}^{\transpose}(i,:) \begin{bmatrix} 1 & -\bm{\phi_{n-1,p-1}}^{\transpose} \end{bmatrix}^{\transpose}| \\
	\medskip  & = \sqrt{\vnorm{\bm{\phi}_{n-1,p-1}}^2 + 1} \left|\XX_{n,p}^{\transpose}(i,:) \frac{\begin{bmatrix} 1 & -\bm{\phi_{n-1,p-1}}^{\transpose} \end{bmatrix}^{\transpose}}{\sqrt{\vnorm{\bm{\phi}_{n-1,p-1}}^2 + 1}}\right| \\	
	\medskip  & \leq \sqrt{\vnorm{\bm{\phi}_{n-1,p-1}}^2 + 1} \vnorm{\XX_{n,p}(i,:)}.
	\end{align*} 
	Now, by using \cref{LemUpBnd4NormX}, we obtain, 
	\begin{align}
	\medskip |\res_{n-1,p-1}(i)| & \leq \sqrt{\vnorm{\bm{\phi}_{n-1,p-1}}^2 + 1} \vnorm{\XX_{n,p}} \sqrt{\ell_{n,p}(i)}.	 	
	\end{align}		
	Inequality \eqref{EquBnd4ErrorHat} can be proved analogously. 
\end{proof}

\begin{lemma}\label{LemBnd4RelError}
	Let $\{y_1,\ldots,y_n\}$ be a time series data. We have,
	\begin{align*}
	\medskip \frac{\left(\res_{n-1,p-1}(i)\right)^{2}}{\vnorm{\res_{n-1,p-1}}^2} & \leq \ell_{n,p}(i),\quad \mbox{for}\ i=1,\ldots,n-p,
	\end{align*}
	where $ \res_{n,p} $ and $\ell_{n,p}(i) $ are defined respectively in \cref{eq:res} and  \cref{def:notation}.
\end{lemma}
\begin{proof}
	Since the leverage score is a non-negative valued function, the proof is directly achieved from \cref{thm:TheRec4LevScr}.  
\end{proof}

\begin{lemma}\label{LemBnd4SSE}
	Let $\{y_1,\ldots,y_n\}$ be a time series data. We have
	\begin{subequations}
		\begin{align}
		\medskip \label{EquBnd4SSE} \vnorm{\res_{n-1,p-1}} & \geq \sqrt{\lambda_{\min}(\XX_{n,p}^{\transpose} \XX_{n,p}) \left(\vnorm{\bm{\phi}_{n-1,p-1}}^2 + 1\right)} , \\
		\medskip \label{EquBnd4SSEHat} \vnorm{\tilde{\res}_{n-1,p-1}} & \geq \sqrt{\lambda_{\min}(\XX_{n,p}^{\transpose} \XX_{n,p}) \left(\vnorm{\tilde{\bm{\phi}}_{n-1,p-1}}^2 + 1\right)} ,
		\end{align}
	\end{subequations}
	where $ \res_{n,p}, \tilde{\res}_{n,p}, \bm{\phi}_{n,p}, \tilde{\bm{\phi}}_{n,p}$, and $\XX_{n,p}$ are defined respectively in \cref{eq:res}, \cref{eq:res_tilde}, \cref{eq:phi}, \cref{eq:phi_tilde}, and \cref{eq:Xnp} and $\lambda_{\min}(.)$ denotes the minimum eigenvalue. 
\end{lemma}
\begin{proof}
	By definition, we have
	\begin{align*}
	\medskip \vnorm{\res_{n-1,p-1}} & = \vnorm{\yy_{n-1,p-1} - \XX_{n-1,p-1}\bm{\phi}_{n-1,p-1}} \\
	\medskip  & = \vnorm{\begin{pmatrix} \yy_{n-1,p-1} & \XX_{n-1,p-1}\end{pmatrix} \begin{bmatrix} 1 & -\bm{\phi}_{n-1,p-1}^{\transpose} \end{bmatrix}^{\transpose}} \\
	\medskip  & = \vnorm{\XX_{n,p} \begin{bmatrix} 1 & -\bm{\phi}_{n-1,p-1}^{\transpose} \end{bmatrix}^{\transpose}} \\
	\medskip  & = \sqrt{\begin{bmatrix} 1 & -\bm{\phi}_{n-1,p-1}^{\transpose} \end{bmatrix} \XX_{n,p}^{\transpose} \XX_{n,p} \begin{bmatrix} 1 & -\bm{\phi}_{n-1,p-1}^{\transpose} \end{bmatrix}^{\transpose}} \\
	\medskip  & \geq \sqrt{\lambda_{\min}(\XX_{n,p}^{\transpose} \XX_{n,p})} \vnorm{\begin{bmatrix} 1 & -\bm{\phi}_{n-1,p-1}^{\transpose} \end{bmatrix}} \\
	\medskip  & = \sqrt{\lambda_{\min}(\XX_{n,p}^{\transpose} \XX_{n,p}) \left(\vnorm{\bm{\phi}_{n-1,p-1}}^2 + 1\right)}.
	\end{align*}
	Inequality \eqref{EquBnd4SSEHat} is proved analogously. 
\end{proof}

\begin{lemma}\label{LemConNumInequ}
	For any positive integer numbers $1<p<n$, we have
	\begin{align*}
	\medskip \kappa(\XX_{n-1,p-1}) & \leq \kappa(\XX_{n,p}),
	\end{align*}
	where $\XX_{n,p}$ is defined in \cref{eq:Xnp} an $\kappa(.)$ denotes the condition number. 
\end{lemma}
\begin{proof}
	It is readily seen that the matrix $\XX_{n,p}$ can be written in the form of 
	\begin{align*}
	\medskip \XX_{n,p} & = \begin{pmatrix}
	\yy_{n,p} & \XX_{n-1,p-1}
	\end{pmatrix}.
	\end{align*}
	On the other hand, by definition, we know that
	\begin{align*}
	\medskip \lambda_{\max}(\XX_{n,p}^{\transpose} \XX_{n,p}) & = \sup_{\vnorm{\bm{\nu}} \leq 1} \bm{\nu}^{\transpose} \XX_{n,p}^{\transpose} \XX_{n,p} \bm{\nu}.
	\end{align*}
	Let $\bm{u}$ be a unit vector corresponding to the maximum eigenvalue $\lambda_{\max}(\XX_{n-1,p-1}^{\transpose} \XX_{n-1,p-1})$ and construct the vector 
	\begin{align*}
	\medskip \bm{\bar{u}} & \defeq \begin{bmatrix}
	0 & \bm{u^{\transpose}}
	\end{bmatrix}^{\transpose}.
	\end{align*}
	We have
	\begin{align*}
	\medskip \lambda_{\max}(\XX_{n,p}^{\transpose} \XX_{n,p}) & \geq \bm{\bar{u}}^{\transpose} \XX_{n,p}^{\transpose} \XX_{n,p} \bm{\bar{u}} \\
	\medskip & = \bm{u}^{\transpose} \XX_{n-1,p-1}^{\transpose} \XX_{n-1,p-1} \bm{u} \\
	\medskip & = \lambda_{\max}(\XX_{n-1,p-1}^{\transpose} \XX_{n-1,p-1}).
	\end{align*}
	Analogously, one can show that $\lambda_{\min}(\XX_{n,p}^{\transpose} \XX_{n,p}) \leq \lambda_{\min}(\XX_{n-1,p-1}^{\transpose} \XX_{n-1,p-1})$. Thus, we have
	\begin{align*}
	\medskip \kappa(\XX_{n,p}) & = \sqrt{\frac{\lambda_{\max}(\XX_{n,p}^{\transpose} \XX_{n,p})}{\lambda_{\min}(\XX_{n,p}^{\transpose} \XX_{n,p})}} \\
	\medskip & \geq \sqrt{\frac{\lambda_{\max}(\XX_{n-1,p-1}^{\transpose} \XX_{n-1,p-1})}{\lambda_{\min}(\XX_{n-1,p-1}^{\transpose} \XX_{n-1,p-1})}} \\
	\medskip & = \kappa(\XX_{n-1,p-1}).
	\end{align*}
\end{proof}

\begin{corollary}\label{CorNormXnoInequ}
	For any positive integer numbers $1<p<n$, we have
	\begin{align*}
	\medskip \vnorm{\XX_{n-1,p-1}} & \leq \vnorm{\XX_{n,p}},
	\end{align*} 
	where $\XX_{n,p}$ is defined in \cref{eq:Xnp}.
\end{corollary}
\begin{proof}
	Since $\lambda_{\max}(\XX_{n,p}^{\transpose} \XX_{n,p}) = \vnorm{\XX_{n,p}}^2$, this inequality is directly derived from the proof of \cref{LemConNumInequ}.
\end{proof}

\subsubsection*{Proof of \cref{thm:quasi}}
\begin{proof}
	By using \cref{thm:TheRec4LevScr,thm:drma}, we have
	\begin{align*}
	\medskip |\ell_{n,p}(i)-\tilde{\ell}_{n,p}(i)| & = \left|\ell_{n-1,p-1}(i) + \frac{\left(\res_{n-1,p-1}(i)\right)^{2}}{\vnorm{\res_{n-1,p-1}}^2} - \ell_{n-1,p-1}(i) - \frac{\left(\tilde{\res}_{n-1,p-1}(i)\right)^{2}}{\vnorm{\tilde{\res}_{n-1,p-1}}^2}\right| \\
	\medskip & =  \left|\frac{\left(\res_{n-1,p-1}(i)\right)^{2}}{\vnorm{\res_{n-1,p-1}}^2} - \frac{\left(\res_{n-1,p-1}(i)\right)^{2}}{\vnorm{\tilde{\res}_{n-1,p-1}}^2} + \frac{\left(\res_{n-1,p-1}(i)\right)^{2}}{\vnorm{\tilde{\res}_{n-1,p-1}}^2} - \frac{\left(\tilde{\res}_{n-1,p-1}(i)\right)^{2}}{\vnorm{\tilde{\res}_{n-1,p-1}}^2}\right| \\
	\medskip & \leq  \left|\frac{\left(\res_{n-1,p-1}(i)\right)^{2}}{\vnorm{\res_{n-1,p-1}}^2} - \frac{\left(\res_{n-1,p-1}(i)\right)^{2}}{\vnorm{\tilde{\res}_{n-1,p-1}}^2}\right| + \left|\frac{\left(\res_{n-1,p-1}(i)\right)^{2}}{\vnorm{\tilde{\res}_{n-1,p-1}}^2} - \frac{\left(\tilde{\res}_{n-1,p-1}(i)\right)^{2}}{\vnorm{\tilde{\res}_{n-1,p-1}}^2}\right| \\
	\medskip & \leq  \left(\res_{n-1,p-1}(i)\right)^{2}\left|\frac{1}{\vnorm{\res_{n-1,p-1}}^2} - \frac{1}{\vnorm{\tilde{\res}_{n-1,p-1}}^2}\right| \\
	\medskip & \ \ \ + \frac{1}{\vnorm{\tilde{\res}_{n-1,p-1}}^2} \left|\left(\res_{n-1,p-1}(i)\right)^{2} - \left(\tilde{\res}_{n-1,p-1}(i)\right)^{2}\right| \\
	\medskip & \\
	\medskip & \leq  \left(\res_{n-1,p-1}(i)\right)^{2}\left|\frac{1}{\vnorm{\res_{n-1,p-1}}^2} - \frac{1}{(1+\varepsilon)^2\vnorm{\res_{n-1,p-1}}^2}\right| \\
	\medskip & \hspace*{0.5cm} + \frac{1}{\vnorm{\res_{n-1,p-1}}^2} \left|\left(\res_{n-1,p-1}(i) - \left(\tilde{\res}_{n-1,p-1}\right)(i)\right)\left(\res_{n-1,p-1}(i) + \left(\tilde{\res}_{n-1,p-1}\right)(i)\right)\right| \\
	\medskip & \leq  \frac{\varepsilon^2+2\varepsilon}{(1+\varepsilon)^2}\frac{\left(\res_{n-1,p-1}(i)\right)^{2}}{\vnorm{\res_{n-1,p-1}}^2} \\
	\medskip & \hspace*{0.5cm} + \frac{1}{\vnorm{\res_{n-1,p-1}}^2} \left|\res_{n-1,p-1}(i) - \left(\tilde{\res}_{n-1,p-1}\right)(i)\right| \\
	\medskip & \hspace*{1.0cm} \times\left(|\res_{n-1,p-1}(i)| + |\left(\tilde{\res}_{n-1,p-1}\right)(i)|\right).
	\end{align*}
	Now, from \cref{LemBnd4DiffError,LemBnd4Error,LemBnd4RelError}, we have
	\begin{align*}
	\medskip |\ell_{n,p}(i)-\tilde{\ell}_{n,p}(i)| & \leq  \frac{\varepsilon^2+2\varepsilon}{(1+\varepsilon)^2}\ell_{n,p}(i) \\
	\medskip & \hspace*{0.5cm} + \frac{1}{\vnorm{\res_{n-1,p-1}}^2} \left(\sqrt{\varepsilon} \eta_{n-1,p-1} \vnorm{\bm{\phi}_{n-1,p-1}} \vnorm{\XX_{n-1,p-1}} \sqrt{\ell_{n-1,p-1}(i)}\right) \\
	\medskip & \hspace*{1.0cm} \times \left(\sqrt{\vnorm{\bm{\phi}_{n-1,p-1}}^2 + 1} \vnorm{\XX_{n,p}} \sqrt{\ell_{n,p}(i)}\right. \\ 
	\medskip & \hspace*{1.5cm} + \left.\sqrt{\vnorm{\tilde{\bm{\phi}}_{n-1,p-1}}^2 + 1} \vnorm{\XX_{n,p}} \sqrt{\ell_{n,p}(i)}\right) \\
	\medskip & \leq \left(\frac{\sqrt{\varepsilon}(2+\varepsilon)}{(1+\varepsilon)^2} + \frac{\eta_{n-1,p-1} \vnorm{\bm{\phi}_{n-1,p-1}} \vnorm{\XX_{n-1,p-1}} \vnorm{\XX_{n,p}}}{\vnorm{\res_{n-1,p-1}}^2} \right. \\
	\medskip & \hspace*{0.5cm} \times \left. \left(\sqrt{\vnorm{\bm{\phi}_{n-1,p-1}}^2 + 1} + \sqrt{\vnorm{\tilde{\bm{\phi}}_{n-1,p-1}}^2 + 1} \right)\right) \sqrt{\varepsilon} \ell_{n,p}(i) \\
	\medskip & \leq \left(1 + \frac{\eta_{n-1,p-1} \vnorm{\bm{\phi}_{n-1,p-1}} \vnorm{\XX_{n-1,p-1}} \vnorm{\XX_{n,p}}}{\vnorm{\res_{n-1,p-1}}^2} \right. \\
	\medskip & \hspace*{0.5cm} \times \left.\left(\sqrt{\vnorm{\bm{\phi}_{n-1,p-1}}^2 + 1} + \sqrt{\vnorm{\tilde{\bm{\phi}}_{n-1,p-1}}^2 + 1} \right)\right) \sqrt{\varepsilon} \ell_{n,p}(i).	
	\end{align*}
	
	\noindent Motivated from \cref{LemBnd4SSE} along with using \cref{CorNormXnoInequ}, we obtain
	\begin{align*}
	\medskip |\ell_{n,p}(i)-\tilde{\ell}_{n,p}(i)| & \leq \left(1 + \frac{\eta_{n-1,p-1} \vnorm{\bm{\phi}_{n-1,p-1}} \vnorm{\XX_{n-1,p-1}} \vnorm{\XX_{n,p}}}{\vnorm{\res_{n-1,p-1}}} \right. \\
	\medskip & \hspace*{0.5cm} \times \left.\left(\frac{\sqrt{\vnorm{\bm{\phi}_{n-1,p-1}}^2 + 1}}{\vnorm{\res_{n-1,p-1}}} + \frac{\sqrt{\vnorm{\tilde{\bm{\phi}}_{n-1,p-1}}^2 + 1}}{\vnorm{\res_{n-1,p-1}}} \right)\right) \sqrt{\varepsilon} \ell_{n,p}(i) \\
	& \leq \left(1 + \frac{\eta_{n-1,p-1} \sqrt{\vnorm{\bm{\phi}_{n-1,p-1}}^2 + 1} \vnorm{\XX_{n,p}}^2}{\vnorm{\res_{n-1,p-1}}} \right. \\
	\medskip & \hspace*{0.5cm} \times \left.\left(\frac{\sqrt{\vnorm{\bm{\phi}_{n-1,p-1}}^2 + 1}}{\vnorm{\res_{n-1,p-1}}} + \frac{(1+\varepsilon)\sqrt{\vnorm{\tilde{\bm{\phi}}_{n-1,p-1}}^2 + 1}}{\vnorm{\tilde{\res}_{n-1,p-1}}} \right)\right) \sqrt{\varepsilon} \ell_{n,p}(i).
	\end{align*}
	Now, by using \cref{LemBnd4SSE}, we obtain
	\begin{align*}
	\medskip |\ell_{n,p}(i)-\tilde{\ell}_{n,p}(i)| & \leq \left(1 + \frac{\eta_{n-1,p-1} \lambda_{\max}(\XX_{n,p}^{\transpose} \XX_{n,p})}{\sqrt{\lambda_{\min}(\XX_{n,p}^{\transpose} \XX_{n,p})}}\right. \\
	\medskip & \hspace*{0.5cm} \times \left.\left(\frac{1}{\sqrt{\lambda_{\min}(\XX_{n,p}^{\transpose} \XX_{n,p})}} + \frac{1+\varepsilon}{\sqrt{\lambda_{\min}(\XX_{n,p}^{\transpose} \XX_{n,p})}} \right)\right) \sqrt{\varepsilon} \ell_{n,p}(i) \\
	\medskip & \leq \left(1 + \frac{3 \eta_{n-1,p-1} \lambda_{\max}(\XX_{n,p}^{\transpose} \XX_{n,p})}{\lambda_{\min}(\XX_{n,p}^{\transpose} \XX_{n,p})}\right) \sqrt{\varepsilon} \ell_{n,p}(i) \\
	\medskip & = \left( 1 + 3 \eta_{n-1,p-1} \kappa^2(\XX_{n,p})\right) \sqrt{\varepsilon} \ell_{n,p}(i).
	\end{align*}
\end{proof}

\subsection{Proof of \cref{thm:main}}

\begin{proof}
	We prove by induction. For $p=2$, it is derived directly from \cref{thm:quasi}. Let us assume that the statement of theorem is correct for all values of $p< \bar{p}$, and prove that it is also correct for $p=\bar{p}$.
	\begin{align*}
	\medskip |\ell_{n,\bar{p}}(i)-\hat{\ell}_{n,\bar{p}}(i)| & = \left|\ell_{n-1,\bar{p}-1}(i) + \frac{\left(\res_{n-1,p-1}(i)\right)^{2}}{\vnorm{\res_{n-1,p-1}}^2} - \hat{\ell}_{n-1,\bar{p}-1}(i) - \frac{\left(\hat{\res}_{n-1,p-1}(i)\right)^{2}}{\vnorm{\hat{\res}_{n-1,p-1}}^2}\right| \\
	\medskip  & \leq \left|\ell_{n-1,\bar{p}-1}(i) - \hat{\ell}_{n-1,\bar{p}-1}(i) \right| +  \left|\frac{\left(\res_{n-1,p-1}(i)\right)^{2}}{\vnorm{\res_{n-1,p-1}}^2} - \frac{\left(\hat{\res}_{n-1,p-1}(i)\right)^{2}}{\vnorm{\hat{\res}_{n-1,p-1}}^2}\right| \\
	\medskip & \leq \left( 1 + 3 \eta_{n-2,\bar{p}-2} \kappa^2(X_{n-1,\bar{p}-1})\right)(\bar{p}-2) \sqrt{\varepsilon} \ell_{n-1,\bar{p}-1}(i) \\
	\medskip & \hspace*{0.5cm} + \left( 1 + 3 \eta_{n-1,\bar{p}-1} \kappa^2(X_{n,\bar{p}})\right) \sqrt{\varepsilon} \ell_{n,\bar{p}}(i) \\
	\medskip & \leq \left( 1 + 3 \eta_{n-1,\bar{p}-1} \kappa^2(X_{n,\bar{p}})\right)(\bar{p}-2) \sqrt{\varepsilon} \ell_{n,\bar{p}}(i) \\
	\medskip & \hspace*{0.5cm} + \left( 1 + 3 \eta_{n-1,\bar{p}-1} \kappa^2(X_{n,\bar{p}})\right) \sqrt{\varepsilon} \ell_{n,\bar{p}}(i) \\
	\medskip & = \left( 1 + 3 \eta_{n-1,\bar{p}-1} \kappa^2(X_{n,\bar{p}})\right)(\bar{p}-1) \sqrt{\varepsilon} \ell_{n,\bar{p}}(i).
	\end{align*}
	The second last inequality comes from the induction hypothesis as well as \cref{thm:quasi} and the last inequality is from \cref{LemConNumInequ}.
\end{proof}

\subsection{Proof of \cref{thm:quality_assurance}}

\begin{proof}
	From \cref{thm:drma}, we have \cref{EquRelDiffError}, which in turn implies
	\begin{align*}
	\medskip \left| \bm{\phi}_{n,p^*}(k)- \hat{\bm{\phi}}_{n,p^*}(k)\right| & \leq \sqrt{\varepsilon} \eta_{n,p^*} \vnorm{\bm{\phi}_{n,p^*}}, \quad \mbox{for}\ 1 \leq k \leq p^*. 
	\end{align*}
	One can estimate the PACF value at lag $ p^* $ using the $ {p^*}\th $ component of the CMLE of the parameter vector based on the full data matrix, i.e., $ \bm{\phi}_{n,p^*}(p^*) $, \cite[Chapter 3]{Shu}. Hence, \cref{eq:tau_1} now readily follows by an application of reverse triangular inequality. 
	
	\noindent To show \cref{eq:tau_2}, we recall that \cite[Chapter 3]{Shu}
	\begin{align*}
	\medskip \mathtt{PACF}_p = \frac{\mathtt{Cov}(p) - \sum_{k=1}^{p-1} \phi_{k} \mathtt{Cov}(p-k)}{\sigma^{2}_{W}},
	\end{align*}
	where $ \mathtt{Cov}(p) $ is the autocovariance function at lag $ p $ and $ \sigma^{2}_{W} $ is the variance of white noise series in an $ \mathtt{AR(p-1)} $ model. It follows that $\tau_{p}$ is given by plugin the CMLE of $ \mathtt{Cov}(p) $, $ \phi_{k}$ for $k = 1,\ldots, p-1 $ and $ \sigma_{W}^2 $, that is, 
	\begin{align*}
	\medskip \tau_p & = \frac{\gamma(p) - \sum_{k=1}^{p-1} \bm{\phi}_{n,p-1}(k) \gamma(p-k)}{\vnorm{\res_{n,p-1}}^{2}/n}.
	\end{align*}
	
	\noindent Hence, for $p > p^*$, we have
	\begin{align*}
	\medskip |\hat{\tau}_p| & = \frac{\left| \gamma(p) - \sum_{k=1}^{p-1} \hat{\bm{\phi}}_{n,p-1}(k) \gamma(p-k) \right|}{\vnorm{\hat{\res}_{n,p-1}}^{2}/n} \\
	\medskip & = \frac{\left| \gamma(p) - \sum_{k=1}^{p-1} \left[\left( \hat{\bm{\phi}}_{n,p-1}(k) + \bm{\phi}_{n,p-1}(k) - \bm{\phi}_{n,p-1}(k) \right) \gamma(p-k) \right] \right| }{\vnorm{\hat{\res}_{n,p-1}}^{2}/n} \\
	\medskip & \leq |\tau_{p} | + \frac{\sum_{k=1}^{p-1} \left|  \left( \bm{\phi}_{n,p-1}(k) - \hat{\bm{\phi}}_{n,p-1}(k) \right) \gamma(p-k) \right|}{\vnorm{\res_{n,p-1}}^{2}/n} \\
	\medskip & \leq |\tau_{p} | + \frac{\gamma(0) \sum_{k=1}^{p-1} \left|  \bm{\phi}_{n,p-1}(k) - \hat{\bm{\phi}}_{n,p-1}(k) \right|}{\vnorm{\res_{n,p-1}}^{2}/n} \\
	\medskip & \leq |\tau_{p} | + \frac{\gamma(0) \sqrt{p-1} \vnorm{\bm{\phi}_{n,p-1} - \hat{\bm{\phi}}_{n,p-1}}}{\vnorm{\res_{n,p-1}}^{2}/n} \\
	\medskip & \leq |\tau_{p} | + \frac{\eta_{n,p} \vnorm{\bm{\phi}_{n,p}} \gamma(0)}{\vnorm{\res_{n,p-1}}^{2}/n}\sqrt{(p-1)\varepsilon}.
	\end{align*}
	Now, the result follows by noting that $ \vnorm{\res_{n,p-1}}^{2}/n $ is an MLE estimate of $ \sigma_{W}^2 $, and from convergence in probability of this estimate, we have that, for large enough $ n $, it is bounded with probability at least $ 1-\delta $.
\end{proof}

\subsection{Proof of \cref{thm:Time Complexity}}

\begin{proof}
	Consider an input $\mathtt{AR}(p^*)$ time series data of size $n$. From \cref{DefHatLS}, \cref{thm:main}, and \cref{rem:ls-time}, given the fully-approximate leverage scores for the data matrix corresponding to the $\mathtt{AR}(p-1)$ models for $p$ varying from $2$ to $p^*$, we can estimate those of $\mathtt{AR}(p)$ models in $ \bigO{n} $ time. Here, we assume that $ \kappa(\XX_{n,p}) $ does not scale with the dimension $ p $ (at least unfavorably so), and treat it as a constant. \cref{thm:main} implies that we must choose $ 0< \varepsilon \leq p^{-2}$.
	Now, solving the compressed OLS problem (e.g., applying QR factorization with Householder reflections) requires $ \mathcal{O} \left( sp^{2} \right) = \mathcal{O} \left( p^{3} \log p/\varepsilon^{2} \right) $. As a result, the overall complexity of performing the \texttt{LSAR} for an input $\mathtt{AR(p^*)}$ time series data is $ \mathcal{O} \left( \sum_{p=1}^{p^*} (n + p^{3} \log p/\varepsilon^{2}) \right) = \mathcal{O} \left( np^* + p^{*^{4}} \log p^*/\varepsilon^{2} \right) $.
\end{proof}


\end{document}